\documentclass[a4paper,12pt]{amsart}
\usepackage{wiaspreprint}

\usepackage[pdftex]{hyperref}

\hypersetup{%
pdftitle = {A new perspective on the Propagation-Separation approach:
Taking advantage of the propagation condition}
pdfauthor = {Saskia Becker, Peter Math\'{e}},
 pdfkeywords = {Structural adaptive smoothing, Propagation, Separation, Local likelihood, Exponential families, Memory},
 pdfstartview = {FitH},
 colorlinks = {false},
 breaklinks = {true},
 pdfborder = {0 0 0},
 linkcolor = {blue},
  plainpages = {false},
}

\usepackage{natbib}
\usepackage{url}
\usepackage{amsmath, amssymb, amsfonts}
\usepackage{color}

\allowdisplaybreaks[3] 

\usepackage{amsthm}
\numberwithin{equation}{section}
\theoremstyle{definition}
\newtheorem{Def}{Definition}[section]
\newtheorem{Alg}{Algorithm}
\newtheorem{Not}[Def]{Notation}
\newtheorem{Ass}{Assumption}
\renewcommand{\theAss}{A\arabic{Ass}}
\theoremstyle{plain}
\newtheorem{Thm}{Theorem}
\newtheorem{Prop}[Def]{Proposition}
\newtheorem{Lem}[Def]{Lemma}
\newtheorem{PS}{PS}

\theoremstyle{remark}
\newtheorem{Cor}[Def]{Corollary}
\newtheorem{Rem}[Def]{Remark}
\newtheorem{Ex}[Def]{Example}

\begin{document}
\author{%
\nofnmark{}				
Saskia Becker,
\footnote{Weierstrass Institute \\
Mohrenstr. 39\\ 10117 Berlin \\ Germany\\
E-Mail: saskia.becker@wias-berlin.de\\
E-Mail: peter.mathe@wias-berlin.de}
Peter Math\'{e} 
}
\title{A new perspective on the Propagation-Separation approach:
Taking advantage of the propagation condition}

\nopreprint{1766}               	
\selectlanguage{english}		
\date{February 5, 2013}			
\subjclass[2010]{ 62G05, 
}	
\keywords{Structural adaptive smoothing,
	 	  Propagation,
		  Separation,
		  Local likelihood,
		  Exponential families}

\begin{abstract}
The Propagation-Separation approach is an iterative procedure for pointwise estimation of local constant and local polynomial functions. The estimator is defined as a weighted mean of the observations with data-driven weights. Within homogeneous regions it ensures a similar behavior as non-adaptive smoothing (propagation), while avoiding smoothing among distinct regions (separation). 
In order to enable a proof of stability of estimates, the authors of the original study introduced an additional memory step aggregating the estimators of the successive iteration steps.
Here, we study theoretical properties of the simplified algorithm, where the memory step is omitted.
In particular, we introduce a new strategy for the choice of the adaptation parameter yielding propagation and stability for local constant functions with sharp discontinuities. 
\end{abstract}

\maketitle

\thispagestyle{empty}

\section{Introduction}

The Propagation-Separation approach~\citep{PoSp05} is an adaptive method for nonparametric estimation. This iterative procedure relates to Lepski's method~\citep{MR1091202, MR2240642} and extends the Adaptive Weights Smoothing (AWS) procedure from~\citet{PoSp00}. 
The Propagation-Separation approach supposes a local parametric model. It is especially powerful in case of large homogeneous regions and sharp discontinuities. However, it can be extended to local linear or local polynomial parameter functions, as well. Hence, the method is applicable to a broad class of nonparametric models. In our study, we concentrate on the local constant model for the sake of simplicity. Important application can be found in image processing, where the local constant model is often satisfied.

In this study, we aim to provide a better understanding of the procedure and its properties. 
The crucial point of the algorithm is the choice of the adaptation bandwidth. 
We present a new formulation of what is known as propagation condition ensuring an appropriate choice. 
This allows the verification of propagation and stability of estimates for local constant parameter functions with sharp discontinuities. 

In comparison to the study of~\citet{PoSp05}, there are two important differences which we want to emphasize. First, we avoid the problematic Assumption S0 on which the theoretical results in~\citep{PoSp05} were partially based. 
Further, we omit the memory step which was included into the algorithm to enable a theoretical study. In each iteration step, the new estimate is compared with the estimate from the previous iteration step. In case of a significant difference the new estimate is replaced by a value between the two estimates, providing a smooth transition, that is relaxation. This is related to the work of \citet{MR2363972} about spatial aggregation of local likelihood estimates 
The theoretical results in \citep{PoSp05} are mainly based on the memory step.
However, we show for piecewise constant functions that the adaptivity of the method yields similar results even if the memory step is removed from the algorithm. This gains importance as it turned out, that for practical use the memory step is questionable. 
Therefore, in later application of the algorithm, the memory step had been omitted, see e.~g. \citet{poas, Li2012, MR2853730, dti_Tabelow08, PoDi08} still yielding the desired behavior in practice. 
This article aims to justify the simplified Propagation-Separation algorithm, where the memory step is removed.

The outline is as follows. After a short introduction of the model and the estimation procedure we introduce a new parameter choice strategy for the adaptation bandwidth. Then, we consider some numerical examples that illustrate the general behavior of the algorithm. The main properties, that is propagation, separation and stability of estimates, will be verified in Section~\ref{sec:theory} for piecewise constant parameter functions with sharp discontinuities. In Section~\ref{sec:discussion}, we justify our new choice of the adaptation bandwidth by analyzing its dependence of the unknown parameter function and by discussing some further questions concerning its application in practice. We finish with a generalization of the setting of our study.

We use two results from~\citet{PoSp05} which do not base on Assumption S0. These are given in Appendix~\ref{app:reminder}.
In order to avoid confusion we refer to them by (PS~\ref{PS 5.2}) and~(PS~\ref{PS 2.1}).

\section{Model and methodology}

In this section we briefly introduce the setting of our study and the estimation procedure resulting from the Propagation-Separation approach. 
The behavior of the algorithm depends on the adaptation bandwidth, and here we introduce a new strategy for its choice.

\subsection{Model}

We consider a local parametric model.

\begin{Not}[Setting]\label{not:setting}
Let $Z_1, ....., Z_n$ be independent random variables with $Z_i = (X_i, Y_i) \in \mathcal{X} \times \mathcal{Y}$. Here, the metric space~$\mathcal{X}$ denotes the design space and $\mathcal{Y} \subseteq \mathbb{R}$ the observation space. The observations~$Y_i$ are assumed to follow the distribution $\mathbb{P}_{\theta(X_i)} \in \mathcal{P}$, where~$\mathcal{P}$ denotes some parametric family of probability distributions and $\theta: \mathcal{X} \to \Theta \subseteq \mathbb{R}$ is the parameter function that we aim to estimate. We suppose the design~$\lbrace X_i \rbrace_{i=1}^n$ to be known.
\end{Not}

Typical examples of this general setting are Gaussian regression or the inhomogeneous Ber\-noulli, Exponential, and Poisson models, see~\citep[Section 2]{PoSp05} for a detailed description. 
In general, the procedure may work for any vector space~$\mathcal{Y} \subseteq M$ with $Y_i \sim \mathbb{P}_{\theta(X_i)}$, $\theta: \mathcal{X} \to \Theta \subseteq M$, where~$M$ is a metric space. 
Following~\citet{PoSp05} we suppose the parametric family to be an exponential family with standard regularity conditions. This allows an explicit expression of the Kullback-Leibler divergence simplifying our following analysis.

\begin{Ass}[Local exponential family model]\label{A1}
$\mathcal{P} =(\mathbb{P}_{\theta}, \theta \in \Theta)$ is an exponential family 
with a compact and convex parameter set~$\Theta$ and non-decreasing functions $C,B \in C^2\left( \Theta, \mathbb{R} \right)$ such that
\[
	p(y, \theta) := d \mathbb{P}_{\theta} / d \mathbb{P} (y) = p(y) \exp \left[ T(y) C(\theta) - B(\theta) \right], \qquad \theta \in \Theta,
\]
where~$p(y)$ is some non-negative function on~$\mathcal{Y}$, $T: \mathcal{Y} \to \mathbb{R} $, and $B'(\theta) = \theta \, C'(\theta)$. For the parameter~$\theta$ it holds
\begin{equation}\label{eq:theta}
	\int p(y, \theta) \mathbb{P}(dy) = 1 \quad \text{ and } \quad 
	\mathbb{E}_{\theta} \left[ T(Y) \right] = \int T(y) p(y, \theta) \mathbb{P}(dy) = \theta.
\end{equation}
\end{Ass}

\begin{Rem}\label{rem:A1}\hspace{1 pt}
\begin{itemize}
\item In \citep[Assumption (A1)]{PoSp05}, the authors assumed $T(y) \equiv y$, i.e. the identity map. 
Any invertible transformation~$T$ leaves the Kullback-Leibler divergence unchanged. 
Since the results (PS~\ref{PS 5.2}) and (PS~\ref{PS 2.1}), see Appendix~\ref{app:reminder}, depend on the Kullback-Leibler divergence only, they remain valid for invertible maps~$T$.
In this study, we consider the general case explicitly in order to clarify, where this transformation $T$ comes into play.
\item Equation~(\ref{eq:theta}), i.e. $\mathbb{E}_{\theta} \left[ T(Y) \right] = \theta$, can be achieved via reparametrization with $\theta := t(\vartheta)$, where $t(\vartheta) := \mathbb{E}_{\vartheta} \left[ T(Y) \right]$. However, this leads to estimation of~$\theta$ instead of~$\vartheta$ such that the theoretical properties in Section~\ref{sec:theory} do not apply for~$\vartheta$. This will be discussed in Section~\ref{sec:genModel}.
\item A list of parametric families satisfying Assumption~(\ref{A1}), probably after reparametrization, is given in Appendix~\ref{app:A1}.
\item We suppose Assumption~(\ref{A1}) throughout this article while all later Assumptions will be required for specific results only.
\end{itemize}
\end{Rem}

In our subsequent analysis the notions of the Kullback--Leibler
divergence, given here as
\[
	\mathcal{KL}\left( \mathbb{P}_{\theta},\mathbb{P}_{\theta'} \right)  
	:= \int \ln \left( \frac{ d( \mathbb{P}_{\theta} ) }{ d( \mathbb{P}_{\theta'} ) } \right) \mathbb{P}_{\theta} (dy), \quad \theta, \theta' \in \Theta,
\]
and the Fisher information
\[
	I (\theta) :=  - \mathbb{E} \left[ \frac{ \partial^2 }{ \partial \theta^2 } \log p(y, \theta) \right], \quad \theta \in \Theta,
\]
will be important.

\begin{Lem}[Fisher information and Kullback-Leibler
  divergence]\label{lem:A1}
Under Assumption~(\ref{A1})
we have that $I(\theta) = C'(\theta),\ \theta\in\Theta$.
Moreover, the following holds.
\begin{itemize}
\item For every constant~$\varkappa \geq 1$ there is a compact  and convex subset $ \Theta_{\varkappa} \subseteq \Theta$ such that 
\begin{equation}\label{eq:varkappa}
	\frac{I(\theta_1)}{I(\theta_2)} \leq \varkappa^2,\quad
        \theta_1, \theta_2 \in \Theta_{\varkappa}.
\end{equation} 
\item The Kullback-Leibler divergence is convex w.r.t. the first argument. It satisfies
\begin{equation}\label{eq:KL}
	\mathcal{KL} \left( \mathbb{P}_{\theta},\mathbb{P}_{\theta'} \right) 
	= \theta \left[ C(\theta) - C(\theta') \right] - \left[ B(\theta) - B(\theta') \right]
	\approx I(\theta) \left[ \theta - \theta' \right]^2 / 2.
\end{equation}
\end{itemize}
\end{Lem}

\begin{proof}[Proof sketch]
The first assertion follows with $B'(\theta) = \theta C'(\theta)$.
Then, Equation~(\ref{eq:varkappa}) holds due to the compactness of~$\Theta_{\varkappa}$ and $C \in C^2 (\Theta, \mathbb{R} )$.
The convexity is satisfied since the second derivative of the Kullback-Leibler divergence is non-negative
\[
	\tfrac{\partial^2}{\partial \theta^2} \, \mathcal{KL} \left( \mathbb{P}_{\theta},\mathbb{P}_{\theta'} \right) = C'(\theta) > 0.
\]
The Taylor expansions of~$B$ and~$C$ yield for the Kullback-Leibler divergence
\begin{eqnarray}
	\mathcal{KL} \left( \mathbb{P}_{\theta}, \mathbb{P}_{\theta'} \right)
	\approx  
	\left[ - \theta C''(\theta) + B''(\theta) \right] ( \theta - \theta' )^2/2 
	= C'(\theta) ( \theta - \theta' )^2/2,\nonumber 
\end{eqnarray}
where $\theta,\theta' \in \Theta$.
\end{proof}

The set~$\Theta_{\varkappa}$ should be sufficiently large such that $\theta(X_i) \in \Theta_{\varkappa}$ holds for all $i \in \lbrace 1,...,n \rbrace$. Later on, we require that even the corresponding estimators are elements of~$\Theta_{\varkappa}$, see Assumption~(\ref{AEst}).
In Remark~\ref{rem:AEst}, we discuss how this can be achieved without increasing~$\varkappa$ overly.

\subsection{Methodology of the Propagation-Separation approach}

The algorithm is iterative, and in each iteration step the pointwise estimator of the parameter function is defined as a weighted mean of the observations. 
In each design point the weights are chosen adaptively as product of two kernel functions.
The \emph{location kernel} acts on the design space~$\mathcal{X}$, and the \emph{adaptation kernel} compares the pointwise parameter estimates of the previous iteration step in terms of the Kullback-Leibler divergence. 
For each of the two kernels, a bandwith controls how much information is taken into account.
The location bandwidth increases along the number of iterations. Starting at a small vicinity, in each iteration step the considered region is extended. 
The increasing number of included observations enables a monotone variance reduction during iteration, while the adaptation kernel leads to a  decreasing or (in case of model misspecification) bounded estimation bias. It will be clear from the subsequent analysis that, by doing so, one obtains similar results as non-adaptive smoothing within homogeneity regions (propagation) and avoids smoothing across structural borders (separation). 

We turn to a formal description, and we start with introducing some notation. 

\begin{Not}\label{not:algorithm}\hspace{1pt}
\begin{itemize}
\item $\theta_i := \theta(X_i)$;
\item $\Delta$ denotes a metric on~$\mathcal{X}$;
\item $\mathcal{KL}(\theta,\theta') := \mathcal{KL}(\mathbb{P}_{\theta},\mathbb{P}_{\theta'})$ is the Kullback-Leibler divergence of~$\mathbb{P}_{\theta}$ and~$\mathbb{P}_{\theta'}$, $\theta, \theta' \in \Theta$;
\item $K_{\mathrm{loc}}, K_{\mathrm{ad}}: \mathbb{R}^+ \to [0,1]$ are non-increasing kernels with compact support~$[0,1]$ and $K_{\cdot}(0) = 1$, where~$K_{\mathrm{loc}}$ denotes the location and~$K_{\mathrm{ad}}$ the adaptation kernel;
\item $\lbrace h^{(k)} \rbrace_{k=0}^{k^*}$ is an increasing sequence of bandwidths for the location kernel with $h^{(0)} > 0$; 
\item $\lambda>0$ is the bandwidth of the adaptation kernel; 
\item $ U_i^{(k)} := \lbrace X_j \in \mathcal{X}: \Delta( X_i, X_j) \leq h^{(k)} \rbrace $.
\end{itemize}
\end{Not}

For comparison and the initialization of the algorithm we define the non-adaptive estimator~$ \overline{\theta}_i^{(k)} $.

\begin{Def}[Non-adaptive estimator]\label{def:MLE}
Let $i \in \lbrace 1,...,n \rbrace$ and $k \in \lbrace 0,...,k^{*} \rbrace$. The non-adaptive estimator~$\overline{\theta}_i^{(k)}$ of~$\theta_i$ is defined by
\[
	\overline{\theta}_i^{(k)} := \sum_{j=1}^n \overline{w}_{ij}^{(k)} T(Y_j) / \overline{N}_i^{(k)}
\] 
with weights $\overline{w}_{ij}^{(k)} := K_{\mathrm{loc}} \left( \Delta(X_i, X_j) / h^{(k)} \right)$, 
and $\overline{N}_i^{(k)} := \sum_j \overline{w}_{ij}^{(k)}$.
\end{Def}

\begin{Cor}[Relation to maximum likelihood estimation]\label{cor:KL}
Assumption~(\ref{A1}) implies that the standard local weighted maximum likelihood estimator
\[
	\theta_i^{(\mathrm{MLE})} := \mathrm{argsup}_{\theta} L(\overline{W}_i^{(k)}, \theta) \quad \text{with} \quad L(\overline{W}_i^{(k)}, \theta) := \sum_j \overline{w}_{ij}^{(k)} \log p(Y_j, \theta),
\]
where $\overline{W}_i^{(k)} := \lbrace \overline{w}_{ij}^{(k)} \rbrace_j$, equals the non-adaptive estimator~$\overline{\theta}_i^{(k)}$ in Definition~\ref{def:MLE}. Further, it follows for the ''fitted log-likelihood'' with $\theta \in \Theta$ that
\[
	L(\overline{W}_i^{(k)}, \theta_i^{(\mathrm{MLE})}, \theta) 
	:= L(\overline{W}_i^{(k)}, \theta_i^{(\mathrm{MLE})}) - L(\overline{W}_i^{(k)}, \theta)
	= \overline{N}_i^{(k)} \mathcal{KL} \left( \overline{\theta}_i^{(k)}, \theta \right).
\]
\end{Cor}

Now, we present the (slightly modified) algorithm of the Propagation-Separation approach allowing $ T(y) \neq y $ and omitting the memory step~\cite[Section 3.2]{PoSp05} by setting $\eta_i \equiv 1$. More details can be found in~\cite[Section 3]{PoSp05}. 

\begin{Alg}[Propagation-Separation approach]\label{algorithm}\hspace{1 pt}
\begin{itemize}
\item Input parameters: Sequence of bandwidths~$\lbrace h^{(k)} \rbrace_{k=0}^{k^*}$ and adaptation bandwidth~$\lambda$.
\item Initialization: $\tilde{\theta}_i^{(0)} := \overline{\theta}_i^{(0)}$ and $\tilde{N}_i^{(0)} := \overline{N}_i^{(0)}$ for all $i \in \lbrace 1,...,n \rbrace$,~$k:= 1$. 
\item Iteration: Do for every $i=1,...,n$
\begin{equation}\label{eq:adEst}
	\tilde{\theta}_i^{(k)} := \sum_{j=1}^n \tilde{w}_{ij}^{(k)} T(Y_j) / \tilde{N}_i^{(k)}
\end{equation}
with weights $\tilde{w}_{ij}^{(k)} := K_{\mathrm{loc}} \left( \Delta(X_i, X_j) / h^{(k)} \right) \cdot K_{\mathrm{ad}} \left( s_{ij}^{(k)} / \lambda \right)$, \\
where $s_{ij}^{(k)} := \tilde{N}_i^{(k-1)} \mathcal{KL}(\tilde{\theta}_i^{(k-1)}, \tilde{\theta}_j^{(k-1)}) $ 
and $\tilde{N}_i^{(k)} := \sum_j \tilde{w}_{ij}^{(k)}$.
\item Stopping: Stop if~$k = k^*$, otherwise increase~$k$ by~$1$.
\end{itemize}
\end{Alg}

\begin{Rem}[Choice of the input parameters]\hspace{1 pt}
\begin{itemize}
\item The amount of adaptivity is determined by the adaptation bandwidth~$\lambda$ which can be specified by the propagation condition independent of the observations at hand, see Sections~\ref{sec:propCond} and~\ref{sec:Lambda} and~\citep[Sections 3.4 and 3.5]{PoSp05}. The choice~$\lambda = \infty$ yields non-adaptive smoothing.
\item The initial location bandwidth~$h^{(0)}$ should be sufficiently small in order to avoid smoothing among distinct homogeneous compartments, before adaptation starts. In practice, any choice of~$h^{(0)}$ such that $U_i^{(0)} = \lbrace X_i \rbrace$ for every $i \in \lbrace 1,...,n \rbrace$ seems to be recommendable. Its drawback is discussed in Remark~\ref{rem:AEst}.
\item The sequence of bandwidth~$\lbrace h^{(k)} \rbrace_{k=0}^{k^*}$ can be chosen such that $h^{(k)} := a^k h^{(0)}$ with $a \approx 1.25^{1/d}$ if~$d$ denotes the dimension of the design space~$\mathcal{X}$, see \citet[Section 3.4]{PoSp05}. Alternatively, we could ensure a constant variance reduction of the estimator, see \citet{poas}.
\item Note, that the procedure provides an intrinsic stopping criterion yielding a certain stability of estimates, see Section~\ref{sec:theory} and the simulations in Figures~\ref{fig:stepF} and~\ref{fig:polynF}. 
Hence, the maximal bandwidth~$h^{(k^*)}$, specified by the maximal number of iterations~$k^*$, is only bounded by the available computation time.
\end{itemize}
\end{Rem}

\subsection{Propagation condition}\label{sec:propCond}

As mentioned above, an appropriate choice of the adaptation bandwidth~$\lambda$ is crucial for the behavior of the algorithm.~\citet[Section 3.5]{PoSp05} suggested a choice, called \emph{propagation condition}. The basic idea is that the impact of the statistical penalty in the adaptive weights should be negligible under homogeneity yielding almost free smoothing within homogeneous regions. More precisely, the authors proposed to adjust~$\lambda$ by Monte-Carlo simulations in accordance with the following criterion, where an artificial data set is considered.
\begin{quote}
"(...) the parameter~$\lambda$ can be selected as the minimal value of~$\lambda$ that, in case of a homogeneous (parametric) model $\theta(x) \equiv \theta$, provides a prescribed probability to obtain the global model at the end of the iteration process."
\end{quote}
Here, we formally introduce a new criterion which allows, in the setting of Algorithm~\ref{algorithm}, the verification of propagation and stability under (local) homogeneity. 
Additionally, it provides a better interpretability than earlier formulations, see e.g. \citet{polzehletal10a}.

Under homogeneity, i.e. if~$\theta(.) \equiv \theta$, (PS~\ref{PS 2.1}) in Appendix~\ref{app:reminder} shows that the non-adaptive estimator satisfies 
$\mathbb{P} \left( \overline{N}_i^{(k)} \mathcal{KL} (\overline{\theta}_i^{(k)}, \theta ) > z \right) \leq 2 e^{-z}$ for all $i \in \lbrace 1,...,n \rbrace$ and every $k \in \lbrace 0,..., k^* \rbrace$.
Hence, $\mathcal{KL} (\overline{\theta}_i^{(k)}, \theta )$ decreases at least with rate~$\overline{N}_i^{(k)}$. 
The following condition ensures a similar behavior for the adaptive
estimator. We introduce the function~$\mathfrak{Z}_{\lambda}: \lbrace 0,...,k^* \rbrace \times (0, 1) \times \Theta \to \mathbb{R}^+$ with~$\lambda > 0$, defined as
\[
	\mathfrak{Z}_{\lambda}(k, p; \theta ) := \inf \left\{ z > 0: \mathbb{P} \left( \overline{N}_i^{(k)} \mathcal{KL}(\tilde{\theta}_i^{(k)}(\lambda), \theta) > z \right) \leq p \right\},
\]
where~$\tilde{\theta}_i^{(k)}(\lambda)$ denotes the adaptive estimator resulting from the Propagation-Separation approach with adaptation bandwidth~$\lambda > 0$ and observations $Y_i \sim \mathbb{P}_{\theta}$ for all $i \in \lbrace 1,...,n \rbrace$, i.e.~$\theta(.) \equiv \theta$.

\begin{Def}[Propagation condition]\label{def:propCond}
We say that~$\lambda$ is chosen in accordance with the propagation condition at level~$\epsilon > 0$ for $\theta \in \Theta$ if the function~$\mathfrak{Z}_{\lambda}(., p; \theta)$ is non-increasing for all $p \in (\epsilon, 1)$.
\end{Def}

As before, the propagation condition is formulated w.r.t. some fixed parameter $\theta \in \Theta$. 
In practice, the parameter function~$\theta(.)$ is unknown. Hence, we need to ensure that the propagation condition is satisfied for all~$\theta_i$ with $i \in \lbrace 1,..., n \rbrace$. At best, the choice of~$\lambda$ by the propagation condition is independent of the underlying parameter~$\theta$.
The study in Section~\ref{sec:Independence} points out that this is the case for Gaussian and exponential distribution and as a consequence for log-normal, Rayleigh, Weibull, and Pareto distribution. Else, we recommend to identify some parameter~$\theta^*$ yielding a sufficiently large choice of the adaptation bandwidth~$\lambda$ such that the propagation condition remains valid for all~$\theta_i$ with $i \in \lbrace 1,..., n \rbrace$, see Section~\ref{sec:Independence} for more details.

\begin{Rem}\label{rem:propCond}\hspace{1 pt}
\begin{itemize}
\item In Section~\ref{sec:Independence}, we consider some examples for Gaussian, exponential and Poisson distribution, see Figures~\ref{fig:Gauss}, \ref{fig:Exp}, and~\ref{fig:Poisson}.
\item If the function~$\mathfrak{Z}_{\lambda}(., p_0, \theta)$, $\theta \in \Theta$, in Definition~\ref{def:propCond} is non-increasing for some $p_0 \in (0, 1)$ then it is non-increasing for all~$p \geq p_0$ by monotonicity.
\item The propagation condition yields a lower bound for the choice of~$\lambda$. In general, it is advantageous to allow as much adaptation as possible without violating the propagation condition. Hence, the optimal choice of~$\lambda$ is
\[
	\lambda_{opt}(\epsilon, \theta) := \inf \left\{ \lambda > 0: \mathfrak{Z}_{\lambda}(.,\epsilon; \theta) \text{ is a non-increasing function} \right\}.
\]
\item In Theorem~\ref{thm:locPropNeuMS} we need~$\epsilon$ to be strictly smaller than~$1/n$. However, this is based on a quite rough upper bound. In practice, it seems advantageous to choose~$\epsilon$ appropriately for the respective application. Note, that~$\lambda_{opt}(\epsilon, \theta)$ increases if~$\epsilon$ decreases.
\item The probability $\mathbb{P} \left( \overline{N}_i^{(k)} \mathcal{KL}(\tilde{\theta}_i^{(k)}(\lambda), \theta) > z \right)$ cannot be calculated exactly. In Section~\ref{sec:approxPropCond}, we introduce an appropriate approximation which can be used in practice. 
\end{itemize}
\end{Rem}

\subsection[Heuristics]{Some heuristic observations}\label{sec:Heurstics}

In order to provide some intuition, we illustrate the general behavior of Algorithm~\ref{algorithm} on two examples, see Figures~\ref{fig:stepF} and~\ref{fig:polynF}. We apply the R-package
\texttt{aws} \citep{aws}. Here, the memory step is skipped by default. It can be included setting \texttt{memory = TRUE}. 

\begin{figure}
\includegraphics[width = 0.32\textwidth]{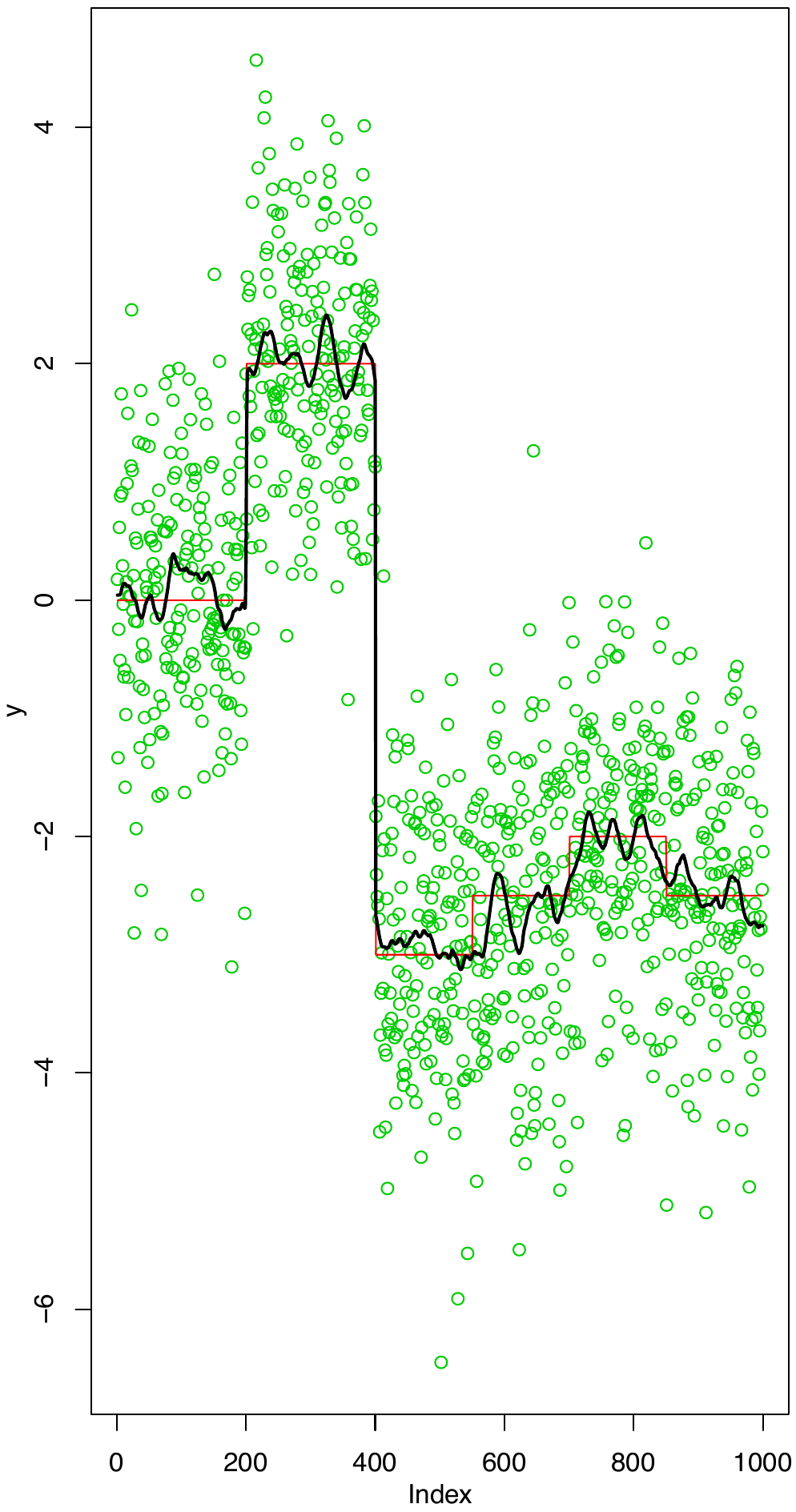} \hspace{1 pt}
\includegraphics[width = 0.32\textwidth]{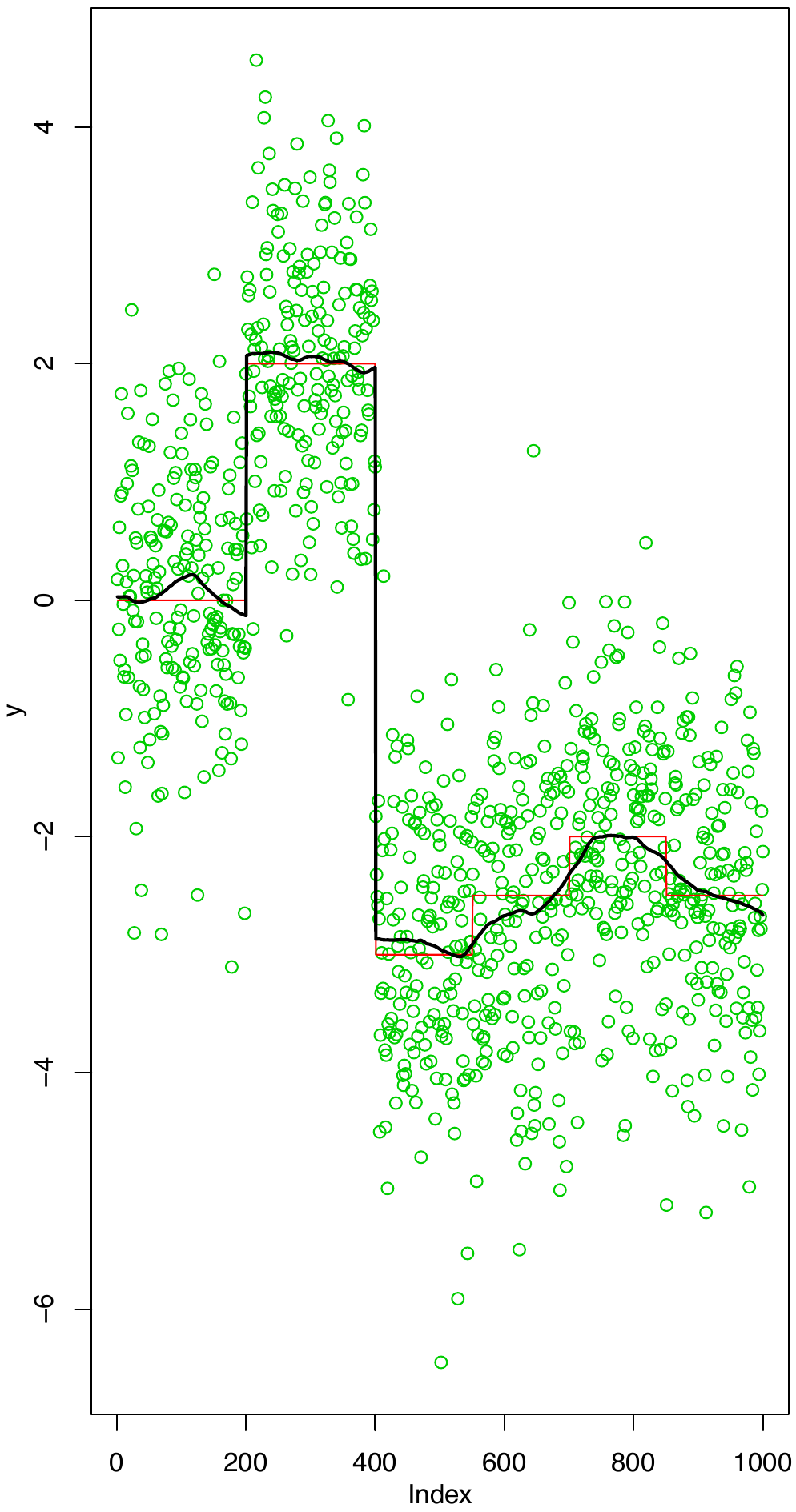} \hspace{1 pt}
\includegraphics[width = 0.32\textwidth]{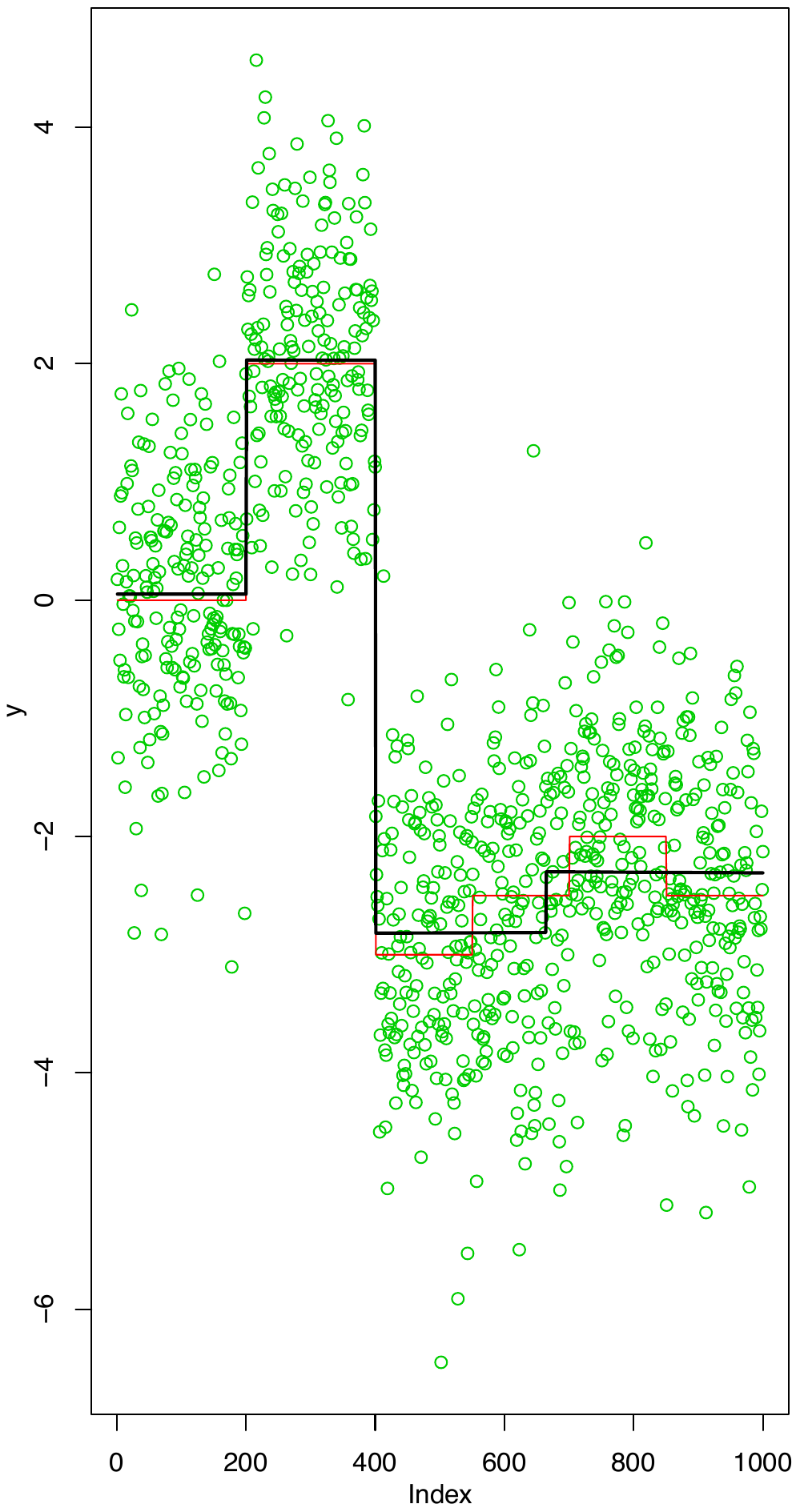}
\caption{{\footnotesize Results of Algorithm~\ref{algorithm} (black line) for the piecewise constant parameter function~$\theta_1(.)$ (red line) with adaptation bandwidth $\lambda_1 = 14.6$ and location bandwidths (f.l.t.r.) $h_1 = 17.1, 52, 947$. The green circles represent the Gaussian distributed observations.}}
\label{fig:stepF}
\end{figure}

On $\mathcal{X} := \lbrace 1,...,1000 \rbrace $, the first test function is piecewise constant
\[
	\theta_1(x) := \begin{cases}
		0, \quad &\text{if } x \in \lbrace 1,...,200 \rbrace\\
		2, \quad &\text{if } x \in \lbrace 201,...,400 \rbrace\\
		-3, \quad &\text{if } x \in \lbrace 401,...,550 \rbrace\\
		-2.5, \quad &\text{if } x \in \lbrace 551,...,700 \rbrace\\
		-2, \quad &\text{if } x \in \lbrace 701,...,850 \rbrace\\
		-2.5, \quad &\text{if } x \in \lbrace 851,...,1000 \rbrace
	        \end{cases}
\]
and the second one is piecewise polynomial
\[
	\theta_2(x) := \begin{cases}
		x / 300, \quad &\text{if } x \in \lbrace 1,...,300 \rbrace\\
		4 + ((x/100-5))^2/2, \quad &\text{if } x \in \lbrace 301,...,800 \rbrace\\
		15 - 2x/100, \quad &\text{if } x \in \lbrace 801,...,1000 \rbrace.
        	\end{cases}
\]
The observations follow a Gaussian distribution, i.e. $Y_i \sim \mathcal{N} \left( \theta(X_i), 1 \right)$. 

\begin{figure}
\includegraphics[width = 0.32\textwidth]{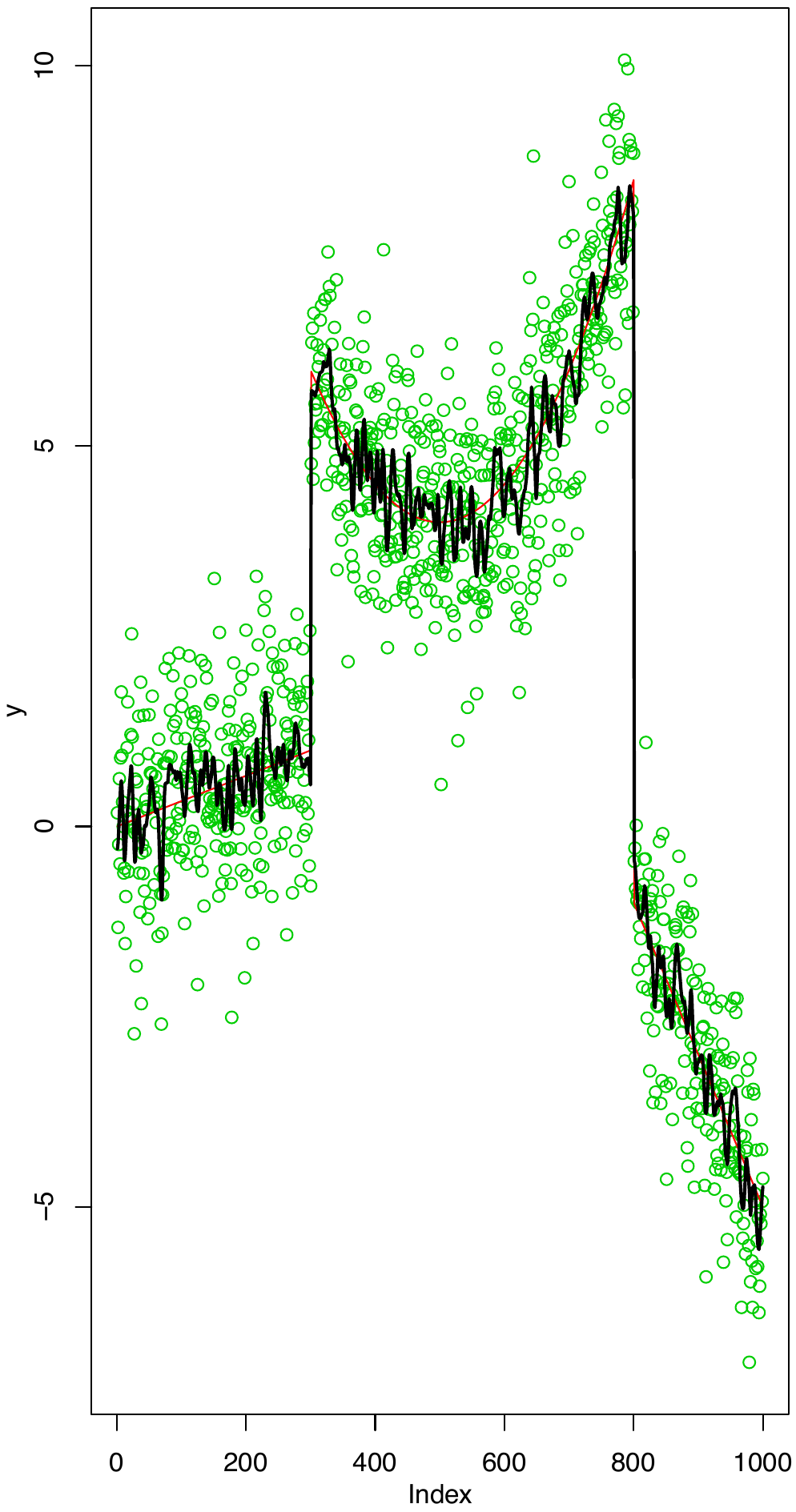} \hspace{1 pt}
\includegraphics[width = 0.32\textwidth]{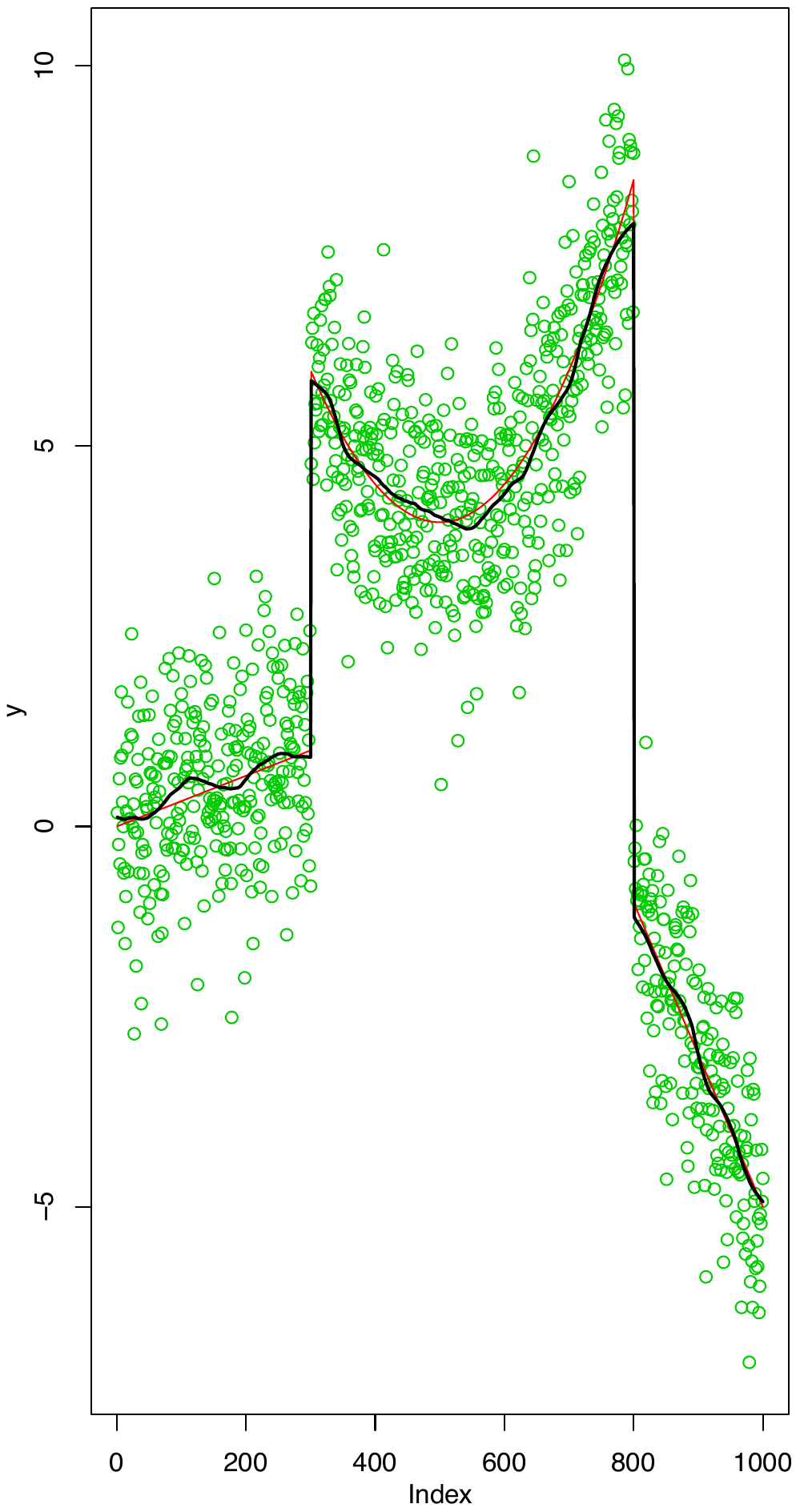} \hspace{1 pt}
\includegraphics[width = 0.32\textwidth]{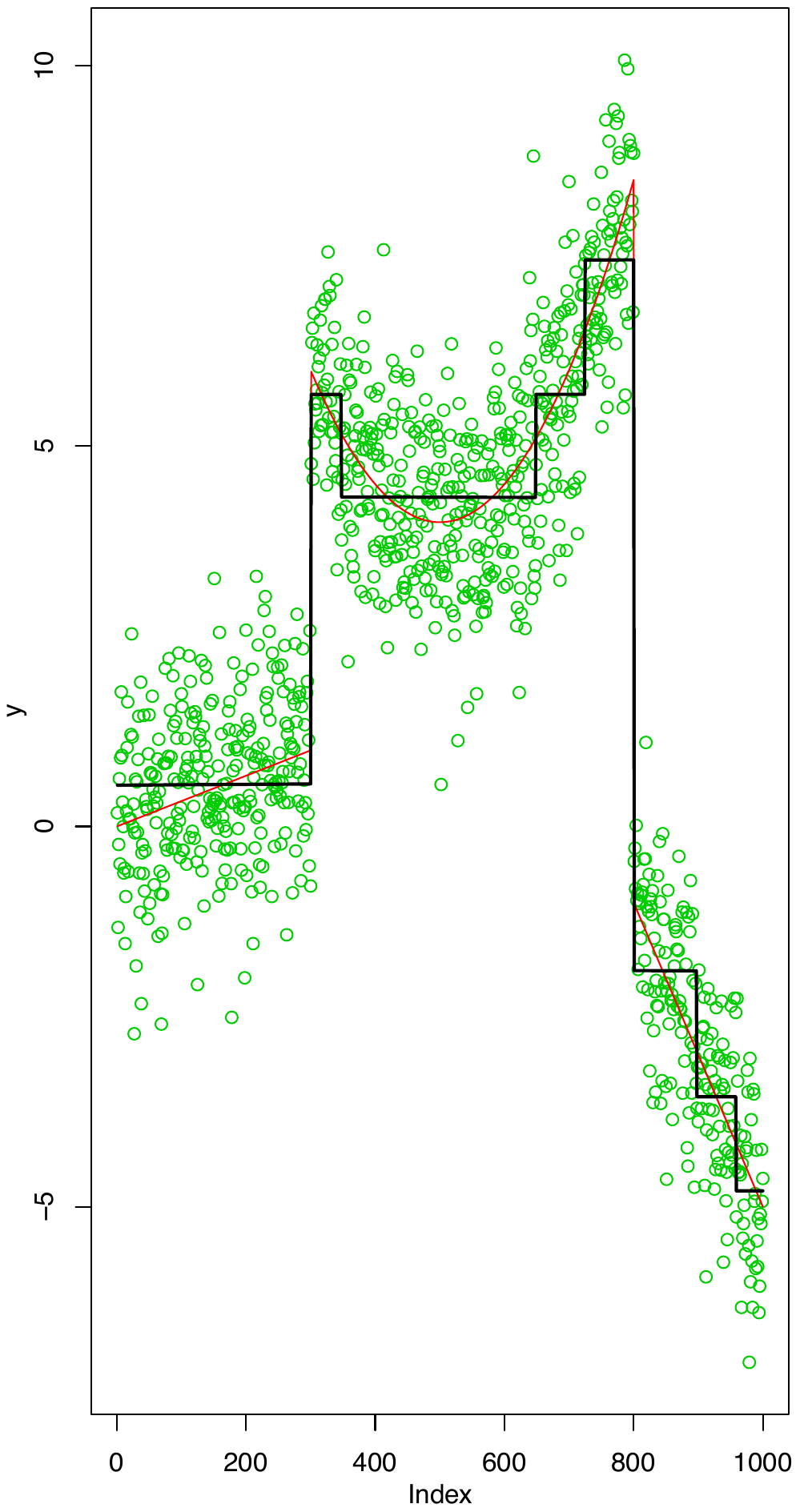}
\caption{{\footnotesize Results of Algorithm~\ref{algorithm} (black line) for the piecewise polynomial parameter function~$\theta_2(.)$ (red line) with adaptation bandwidth $\lambda_2 = 16$ and location bandwidths (f.l.t.r.) $h_2 = 4.42, 41.6, 947$. The green circles correspond to the Gaussian distributed observations. }}
\label{fig:polynF}
\end{figure}

The plots were provided by the function \texttt{aws} setting $\texttt{hmax} :=h^{(k^*)} := 1000$ and \texttt{lkern = "Triangle"}, such that
\begin{equation}\label{eq:kernels}
	K_{\mathrm{loc}}(x) := 1-x^2 
	\quad \text{ and } \quad
	K_{\mathrm{ad}}(x) := \min \lbrace 1,2-x \rbrace_+.
\end{equation}
In Figure~\ref{fig:stepF}, we show the results for the piecewise constant function~$\theta_1(.)$ with $\lambda_1 = 14.6$ and increasing location bandwidths $h_1 = 17.1, 52, 947$ corresponding to the iteration steps $k_1 = 15, 20, 33$. Figure~\ref{fig:polynF} is based on the piecewise smooth function~$\theta_2(.)$ setting $\lambda_2 = 16$ and $h_2 = 4.42, 41.6, 947$, that is $k_2 = 9, 19, 33$. For both examples, it holds $k^* = 33$ representing the final iteration step. The corresponding mean squared error (MSE) is similar to the MSE in step $k_1 =15$ and $k_2 = 9$, respectively. In the steps $k_1 = 20$ an $k_2 = 19$ the MSE is minimal.

We summarize the following heuristic observations.
\begin{itemize}
\item Homogeneous compartments with sufficiently large discontinuities are separated by the algorithm leading to a consistent estimator, see $x \in \lbrace 1,..., 400 \rbrace$ in Figure~\ref{fig:stepF}.
\item If the discontinuities are too small, separation fails. Then, different homogeneous compartments are treated as one yielding a bounded estimation bias. This is illustrated in the right part of Figure~\ref{fig:stepF}, where $x \in \lbrace 401,..., 1000 \rbrace$.
\item In Figure~\ref{fig:polynF}, we consider the case of model misspecification, that is a parameter function~$\theta(.)$ that is not piecewise constant. Here, the algorithm forces the final estimator into a step function. The step size depends mainly on the smoothness of the parameter function~$\theta(.)$ and the adaptation bandwidth~$\lambda$.
However, the estimation bias can be reduced by an accurate stopping criterion. 
The maximal location bandwidth $h^{(k^*)}$ should be chosen such that the non-adaptive estimator in Definition~\ref{def:MLE} behaves good within regions without discontinuities. Then, supposing an appropriate choice of the adaptation bandwidth~$\lambda$, within these regions, Algorithm~\ref{algorithm} would yield similar results as non-adaptive smoothing while smoothing among distinct regions would be avoided as sharp discontinuities could be detected by the adaptive weights. 
\end{itemize}


Thus, the heuristic properties are quite clear. However, the iterative approach complicates a theoretical verification considerably. Therefore, in Section~\ref{sec:theory} we concentrate on piecewise constant functions with sharp discontinuities. Here, our new propagation condition, see Section~\ref{sec:propCond} ensures propagation within homogeneous regions and stability of estimates due to separation of distinct compartments. The case of model misspecification will be analyzed in an upcoming study.

\section{Theoretical properties}\label{sec:theory}

Now, we analyze the behavior of the algorithm in more detail. First, we consider a homogeneous setting, where propagation and stability of estimates follow as direct consequence of the propagation condition. Then, we show the separation property. For locally constant parameter functions with sufficiently sharp discontinuities this restricts smoothing to the respective homogeneous regions yielding again propagation and a certain stability of estimates. We assume that we have identified~$\lambda$ and~$\epsilon$ such that the propagation condition holds.

\subsection[Homogeneity]{Propagation and stability under homogeneity}\label{sec:PropMS}

We show for a homogeneous setting that the propagation condition yields with (PS~\ref{PS 2.1}) in Appendix~\ref{app:reminder} an exponential bound for the excess probability $ \mathbb{P} \left( \overline{N}_i^{(k)} \mathcal{KL}(\tilde{\theta}_i^{(k)}, \theta) > z \right)$ of the Kullback-Leibler divergence between the adaptive estimator~$ \tilde{\theta}_i^{(k)} $ and the true parameter~$ \theta $.

\begin{Prop}[Propagation and stability under homogeneity]\label{prop:propCond}
Suppose $\theta(.) \equiv \theta$, Assumption~(\ref{A1}), and let the
adaptation bandwidth~$\lambda$ be chosen in accordance with the
propagation condition at level~$\epsilon$ for $\theta \in
\Theta$. Then, for each $i \in \lbrace 1,...,n \rbrace$, $k \in \lbrace 0,...,k^* \rbrace$, and all~$z>0$, it holds 
\begin{equation}\label{eq:propCond1}
	\mathbb{P} \left( \overline{N}_i^{(k)} \mathcal{KL} \left( \tilde{\theta}_i^{(k)}, \theta \right) > z \right) 
	\leq \max \left\{ 2 e^{-z}, \epsilon \right\}.
\end{equation}
In particular, we get for all~$k' \geq k$ that
\begin{equation}\label{eq:propCond2}
	\mathbb{P} \left( \overline{N}_i^{(k')} \mathcal{KL} \left( \tilde{\theta}_i^{(k')}, \theta \right) > z \right) \leq \max \left\{ \mathbb{P} \left( \overline{N}_i^{(k)} \mathcal{KL} \left( \tilde{\theta}_i^{(k)}, \theta \right) > z \right), \epsilon \right\}.
\end{equation}
\end{Prop}

\begin{proof}
Equation~(\ref{eq:propCond2}) follows from the propagation condition,
which ensures that the function~$\mathfrak{Z}_{\lambda}(., p; \theta)$
is non-increasing for all~$p \in (\epsilon, 1)$.
Since, see  Algorithm~\ref{algorithm}, we have~$\tilde{\theta}_i^{(0)} = \overline{\theta}_i^{(0)} $ this yields
\begin{eqnarray*}
	\mathbb{P} \left( \overline{N}_i^{(k)} \mathcal{KL} \left( \tilde{\theta}_i^{(k)}, \theta \right) > z \right) 
	&\overset{\text{Eq.(\ref{eq:propCond2})}}{\leq}&
	\max \left\{ \mathbb{P} \left( \overline{N}_i^{(0)} \mathcal{KL} \left( \overline{\theta}_i^{(0)}, \theta \right) > z \right), \epsilon \right\}\\
	&\overset{\text{(PS~\ref{PS 2.1})}}{\leq}& \max \left\{ 2 e^{-z}, \epsilon \right\},
\end{eqnarray*}
leading to the assertion.
\end{proof}

\subsection{Separation property}\label{sec:separation}

For considerably different parameter values the corresponding adaptive weights become zero, see Proposition~\ref{prop:separationMS} below.
To show this, we need (PS~\ref{PS 5.2}) in Appendix~\ref{app:reminder}. 
This requires an appropriate choice of the constant~$\varkappa > 0$, introduced in Lemma~\ref{lem:A1}. The iteration step $k \in \lbrace 0,...,k^* \rbrace$ will be specified in each case where the assumption is used.

\begin{Ass}[Choice of~$\varkappa$]\label{AEst}
Let~$\varkappa > 0$ be sufficiently large such that the true parameter and its estimator satisfy $\theta_i, \tilde{\theta}_i^{(k)} \in \Theta_{\varkappa}$ for all $i \in \lbrace 1,...,n \rbrace$.
\end{Ass}

\begin{Rem}\label{rem:AEst}
Suppose that~$\varkappa$ satisfies $\theta_i \in \Theta_{\varkappa}$ for all $i \in \lbrace 1,...,n \rbrace$.
Then it holds with high probability, for sufficiently large iteration steps~$k$, that $\tilde{\theta}_i^{(k)} \in \Theta_{\varkappa}$, too. However, in Theorem~\ref{thm:locPropNeuMS} we require Assumption~(\ref{AEst}) for all iteration steps. 
In order to ensure this, we could increase~$\varkappa$ leading to a larger set~$\Theta_{\varkappa}$, but this would weaken our theoretical results.
Instead, we recommend a slight modification of the algorithm. 
We replace Equation~(\ref{eq:adEst}) by
\[
	\tilde{\theta}_i^{(k)} := 
	\underset{\theta' \in \Theta_{\varkappa}}{\mathrm{argmin}} \, \left| \theta' - \sum_{j=1}^n \tilde{w}_{ij}^{(k)} Y_j / \tilde{N}_i^{(k)} \right|,
\]
projecting the adaptive estimator into the set~$\Theta_{\varkappa}$. 
This approach corresponds to Bayesian estimation with a priori knowledge $\theta_i \in \Theta_{\varkappa}$ for all $i \in \lbrace 1,...,n \rbrace$. 
Analogously, we redefine the initial estimates via projection of the non-adaptive estimator into~$\Theta_{\varkappa}$
\[
	\tilde{\theta}_i^{(0)} := \underset{\theta' \in \Theta_{\varkappa}}{\mathrm{argmin}} \, \left| \theta' - \overline{\theta}_i^{(0)} \right|.
\]
Additionally, it might be advantageous to decrease the probability of $\overline{\theta}_i^{(0)} \notin \Theta_{\varkappa}$ by choosing the initial bandwidth~$h^{(0)}$ such that the neighborhood~$U_i^{(0)}$ contains more design points than~$X_i$ for each $i \in \lbrace 1,...,n \rbrace$. Else, the projection may change the adaptive weights in later iteration steps leading to slightly shifted estimators. On the other hand, initialization with $U_i^{(0)} = \lbrace X_i \rbrace$ avoids smoothing among distinct homogeneous regions before adaptation starts. 
\end{Rem}

The following proposition is similar to the first part of~\cite[Theorem 5.9]{PoSp05}.
It implies that different homogeneous compartments with sufficiently large discontinuities 
will be separated by the algorithm. 
In particular, we will see, that the lower bound for the discontinuities allowing exact separation of the distinct compartments depends mainly on the adaptation bandwidth~$\lambda$ and the achieved quality of estimation in the previous iteration step.

\begin{Prop}[Separation property]\label{prop:separationMS}
Suppose Assumptions~(\ref{A1}) and, at iteration step~$k$, Assumption~(\ref{AEst}). 
We consider two points~$X_{i_1}$ and~$X_{i_2}$ providing in iteration step~$k$ the estimation accuracy $\mathcal{KL}(\tilde{\theta}_{i_m}^{(k)}, \theta_{i_m}) \leq z_m^{(k)} := z / \overline{N}_{i_m}^{(k)} $ with some constant~$z > 0$, $m=1,2$. 
If 
\begin{equation}\label{eq:varphi1}
	\mathcal{KL}^{1/2} \left( \theta_{i_1}, \theta_{i_2} \right) > \varkappa \left( \sqrt{\lambda /\tilde{N}_{i_1}^{(k)}} + \sqrt{z_1^{(k)}} + \sqrt{z_2^{(k)}} \right)
\end{equation}
then it holds~$\tilde{w}_{i_1 i_2}^{(k+1)} = 0$.
\end{Prop}
 
\begin{proof}[Proof sketch]
Due to the compact support of the adaptation kernel~$K_{\mathrm{ad}}$, it suffices to show that the statistical penalty introduced in Algorithm~\ref{algorithm} satisfies $ s_{i_1 i_2}^{(k+1)} > \lambda $.
(PS~\ref{PS 5.2}) in Appendix~\ref{app:reminder} yields for $ \mathcal{KL}(\tilde{\theta}_{i_m}^{(k)}, \theta_{i_m}) \leq z_m^{(k)}$ with~$ m = 1,2 $ that
\begin{eqnarray*}
	\mathcal{KL}^{1/2} \left( \tilde{\theta}_{i_1}^{(k)}, \tilde{\theta}_{i_2}^{(k)} \right)
	\overset{\text{(\ref{AEst})}}{\geq} \varkappa^{-1} \mathcal{KL}^{1/2} \left( \theta_{i_1}, \theta_{i_2} \right) 
	- \sqrt{z_1^{(k)}} - \sqrt{z_2^{(k)}}
\end{eqnarray*}
such that
\begin{eqnarray*}
	s_{i_1 i_2}^{(k+1)} 
	\geq \tilde{N}_{i_1}^{(k)} \left[ \varkappa^{-1} \sqrt{ \mathcal{KL} \left( \theta_{i_1}, \theta_{i_2} \right) }
	- \sqrt{z_1^{(k)}} - \sqrt{z_2^{(k)} } \right]^2	> \lambda,
\end{eqnarray*}
by Equation~(\ref{eq:varphi1}).
\end{proof}

\begin{Rem}
The lower bound~(\ref{eq:varphi1}) holds if
\[
	\mathcal{KL}^{1/2} \left( \theta_{i_1}, \theta_{i_2} \right) > 3 \varkappa \cdot \frac{ \max \left\{ \sqrt{\lambda }, \sqrt{z} \right\} }{ \min \left\{ \sqrt{ \tilde{N}_{i_1}^{(k)} }, \sqrt{ \overline{N}_{i_1}^{(k)} }, \sqrt{\overline{N}_{i_2}^{(k)} } \right\} }.
\]
This emphasizes the impact of the involved sample sizes.
\end{Rem}

\subsection[Local homogeneity]{Propagation and stability under local homogeneity} \label{sec:specialCase}

Next, we consider a locally homogeneous setting with sharp discontinuities. 
In this case, smoothing is restricted to the homogeneous compartments leading to similar results as under homogeneity, that is to propagation and to stability of estimates.

\begin{Ass}[Structural assumption]\label{AS}
There is a non-trivial partition $\mathcal{V}:= \lbrace \mathcal{V}_i \rbrace_i$ of~$\mathcal{X}$ into maximal homogeneity compartments, i.e. for each $X_{i} \in \mathcal{X}$ there are a vicinity $\mathcal{V}_{i} \subseteq \mathcal{X}$ and a constant~$\varphi_{i} > 0$ such that
\[
\begin{cases}
	\mathcal{KL} \left( \theta_{i}, \theta_j \right) = 0 &\text{ for all } X_j \in \mathcal{V}_{i}\\ 
	\mathcal{KL} \left( \theta_{i}, \theta_j \right) > \varphi_{i}^2 &\text{ for all } X_j \notin \mathcal{V}_{i}.
\end{cases}
\]
\end{Ass}

We deduce the propagation property for the present case. Here, we should take into account that the considered neighborhood~$U_i^{(k)}$ might be much larger than the respective homogeneity compartment~$\mathcal{V}_i$. Obviously, the divergence $\mathcal{KL} ( \tilde{\theta}_i^{(k)}, \theta_i ) $ cannot converge with rate~$\overline{N}_i^{(k)}$ in this case. 
Therefore, we introduce the notion of the effective sample size~$\overline{n}_i^{(k)}$.

\begin{Not}\label{not:effectiveSample}
We define for each $i \in \lbrace 1,...,n \rbrace$ and $ k \in \lbrace 0,..., k^* \rbrace $ the effective sample size and its local minimum
\begin{equation}\label{eq:effectiveSample}
	\overline{n}_i^{(k)} := 
	\sum_{X_{j} \in \mathcal{V}_{i} \cap U_{i}^{(k)}} \overline{w}_{ij}^{(k)}
	\qquad \text{ and } \qquad
	n_i^{(k)} := \underset{X_{j} \in U_i^{(k)} }{\min} \overline{n}_j^{(k)}.
\end{equation}
\end{Not}

As it turns out, the quantities $n_i^{(k)}$ determine the minimal
stepsizes $\varphi_{i}$ such that a discontinuity will be detected.
During the first iteration steps it holds~$\overline{n}_i^{(k)} = \overline{N}_i^{(k)}$. The quotient $ \overline{n}_i^{(k)} / \overline{N}_i^{(k)} $ decreases when~$U_i^{(k)}$ becomes larger than~$\mathcal{V}_i$.

In the following theorem, we consider the event
\[
	\mathcal{B}^{(k)} (z) := \left\{ \overline{n}_i^{(k)} \mathcal{KL} ( \tilde{\theta}_i^{(k)}, \theta_i ) \leq z \quad \text{for all } i  \right\}, \qquad z > 0.
\]

\begin{Thm}[Propagation property under local homogeneity]\label{thm:locPropNeuMS}
Suppose Assumptions~(\ref{A1}) and (\ref{AS}), and, for all iteration steps~$k < k'$ with $k' \in \lbrace 0,..., k^* \rbrace$ fixed, Assumption~(\ref{AEst}). 
Let the bandwidth~$\lambda$ be chosen in accordance with the propagation condition at level~$\epsilon$ for all~$\theta_i$, $i \in \lbrace 1,...,n \rbrace$.
If for all $i \in \lbrace 1,...,n \rbrace$ and every~$k < k'$ the constants~$\varphi_{i} > 0$ in Assumption~(\ref{AS}) satisfy
\begin{equation}\label{eq:varphi3}
	\varphi_{i} > \varkappa \left[ \sqrt{\lambda/ \tilde{N}_{i}^{(k)} } + 2 \sqrt{ z/ n_i^{(k)} } \right],
\end{equation}
then
\begin{equation}\label{eq:locPropNeuMS2}
	\mathbb{P} \left( \mathcal{B}^{(k')}(z) \right) \geq 1 - (k'+1) \, \max \left\{ 2 n e^{-z }, n \epsilon \right\}. 
\end{equation}
\end{Thm}

\begin{proof}
Let~$M^c$ denote the complement of the set~$M$.
Then it holds
\begin{align}
	\mathbb{P} \left( \mathcal{B}^{(k)}(z) \right)
	&=1 - \mathbb{P} \left( (\mathcal{B}^{(k)}(z) )^c \cap \mathcal{B}^{(k-1)}(z) \right) 
	- \mathbb{P} \left( (\mathcal{B}^{(k)}(z) )^c \cap (\mathcal{B}^{(k-1)}(z))^c \right) \nonumber\\
	&\geq  1 - n \cdot \mathbb{P} \left( \lbrace \overline{n}_{i}^{(k)} \mathcal{KL} \left( \tilde{\theta}_{i}^{(k)}, \theta_{i} \right) > z \rbrace \cap \mathcal{B}^{(k-1)}(z) \right) 
	- \mathbb{P} \left( (\mathcal{B}^{(k-1)}(z))^c \right). \label{eq:Bkmu}
\end{align}
Due to~(\ref{eq:varphi3}) the conditions of Proposition~\ref{prop:separationMS} are satisfied on~$\mathcal{B}^{(k-1)}(z)$. Therefore, it follows  on~$\mathcal{B}^{(k-1)}(z)$ that~$\tilde{w}_{ij}^{(k)} = 0$ for all $ X_{j} \notin U_{i}^{(k)} \cap \mathcal{V}_{i}$. Hence, smoothing is restricted to the homogeneous compartment~$\mathcal{V}_{i}$ and $ \mathbb{E} \tilde{\theta}_{i}^{(k)} = \theta_{i} $. 
We get with Proposition~\ref{prop:propCond}
\begin{eqnarray}
	\mathbb{P} \left( \lbrace \overline{n}_{i}^{(k)} \mathcal{KL} \left( \tilde{\theta}_{i}^{(k)}, \theta_{i} \right) > z \rbrace \cap \mathcal{B}^{(k-1)}(z) \right) 
	\leq 	\max \left\{ 2 e^{- z }, \epsilon \right\} \label{eq:locPropNeuMS3}
\end{eqnarray}
 for all~$k \in \lbrace 1,..., k' \rbrace $. Now, we proceed by induction. 
Since $\tilde{\theta}_{i}^{(0)} = \overline{\theta}_{i}^{(0)}$ by Algorithm~\ref{algorithm} it follows from (PS~\ref{PS 2.1}) in Appendix~\ref{app:reminder} that
\begin{eqnarray*}
	\mathbb{P} \left( \mathcal{B}^{(0)} (z)\right)
	\overset{\overline{n}_{i}^{(0)} \leq \overline{N}_{i}^{(0)}}{\geq}  1 - n \cdot \mathbb{P} \left( \left\{ \overline{N}_{i}^{(0)} \mathcal{KL} ( \overline{\theta}_{i}^{(0)}, \theta_{i} ) > z \right\} \right)
	\overset{\text{(PS~\ref{PS 2.1})}}{\geq} 1 - 2 \, n e^{- z}.
\end{eqnarray*}
Finally, Equations~(\ref{eq:Bkmu}) and~(\ref{eq:locPropNeuMS3}) lead for all~$k \leq k'$ to
\begin{eqnarray*}
	\mathbb{P} \left( \mathcal{B}^{(k)}(z) \right)
	&\geq&  1 - n \max \left\{ 2 e^{- z}, \epsilon \right\}  - k \max \left\{ 2 n e^{- z }, n \epsilon \right\}\\
	&=& 1 - (k+1) \max \left\{ 2 n e^{- z }, n \epsilon \right\}.
\end{eqnarray*}
This terminates the proof.
\end{proof}

\begin{Rem}\hspace{0 pt}
\begin{itemize}
\item In Equation~(\ref{eq:locPropNeuMS2}), we observe an additional
  factor~$(k+1)$, which appeared in the propagation property of
  \citet{PoSp05} as well, see Equation~(\ref{eq:propPS06}) in Section~\ref{sec:comparison}, below. This
  factor results from the proof only and might be avoidable. In
  particular, we notice that the given bound is not sharp as we did
  not take advantage of the intersections of the sets $ \left(
    \mathcal{B}^{(k)}(z) \right)^c $ in Equation~(\ref{eq:Bkmu}). The
  above theorem provides a meaningful result for $z \geq q \log(n)$
  and $\epsilon := c_{\epsilon} n^{-q}$ with~$c_{\epsilon} > 0$ and~$q
  > 1$.
\item Separation depends via the statistical penalty on the estimation quality of all data within the local neighborhood~$U_i^{(k)}$. Therefore, the extension of the smallest homogeneous compartment, denoted by~$n_i^{(k)}$, determines the lower bound~(\ref{eq:varphi3}) for the discontinuities that provide an exact separation of the distinct homogeneous compartments. This bound is closely related to Equation~(\ref{eq:varphi1}) that involves only two points  such that
the term $2/ \sqrt{ n_i^{(k)} }$ from Equation~(\ref{eq:varphi3}) can be replaced by
\[
	\left( 1 / \sqrt{\overline{N}_{i_1}^{(k)}} + 1 / \sqrt{\overline{N}_{i_2}^{(k)}} \right)
\]
having the same effect.
\end{itemize}
\end{Rem}

%

Finally, we deduce a similar result as in Equation~(\ref{eq:propCond2}) under local homogeneity. 
Thus, we infer from the estimation quality in iteration step $k_1$ on the estimation quality in step $k_2 > k_1$.
To this end, we apply again the separation property, see Proposition~\ref{prop:separationMS}.
This requires sure knowledge on the previously achieved estimation quality. 
Therefore, we consider the conditional probability and verify an exponential bound.

\begin{Prop}[Stability of estimates under local homogeneity] \label{prop:locStability}
In the situation of Theorem~\ref{thm:locPropNeuMS}, it holds for all $k_1, k_2 \in \lbrace 0,..., k^* \rbrace$ with $k_1 < k_2 \leq k'$ such that $ (k_2 + 1) \, \max \left\{ 2 n e^{-z }, n \epsilon \right\} < 1 $ that
\begin{equation}\label{eq:locStability}
	\mathbb{P} \left( \mathcal{B}^{(k_2)} (z) \vert \mathcal{B}^{(k_1)} (z) \right) 
	\geq \frac{1 - (k_2 + 1) \, \max \left\{ 2 n e^{-z }, n \epsilon \right\}}{1 - (k_1 + 1) \, \max \left\{ 2 n e^{-z }, n \epsilon \right\}}
\end{equation}
\end{Prop}

\begin{proof}
The lower bound holds since
\[
	\mathbb{P} \left( \mathcal{B}^{(k_2)} (z) \vert \mathcal{B}^{(k_1)} (z) \right)
	= 1 - \frac{\mathbb{P} \left( (\mathcal{B}^{(k_2)} (z))^c \cap \mathcal{B}^{(k_1)} (z) \right)}{\mathbb{P} \left( \mathcal{B}^{(k_1)} (z) \right)}
\]
and furthermore 
\begin{eqnarray*}
	& & \mathbb{P} \left( (\mathcal{B}^{(k_2)} (z))^c \cap \mathcal{B}^{(k_1)} (z) \right)\\
	& = & \mathbb{P} \left( (\mathcal{B}^{(k_2)} (z))^c \cap \mathcal{B}^{(k_2 - 1)} (z) \cap \mathcal{B}^{(k_1)} (z) \right) \\
	&& + \mathbb{P} \left( (\mathcal{B}^{(k_2)} (z))^c \cap (\mathcal{B}^{(k_2 - 1)} (z))^c \cap \mathcal{B}^{(k_1)} (z) \right)\\
	& \leq & \mathbb{P} \left( (\mathcal{B}^{(k_2)} (z))^c \cap \mathcal{B}^{(k_2 - 1)} (z) \right)
	+ \mathbb{P} \left( (\mathcal{B}^{(k_2 - 1)} (z))^c \cap \mathcal{B}^{(k_1)} (z) \right)\\
	& \leq & \sum_{k = k_1 + 1}^{k_2} \mathbb{P} \left( (\mathcal{B}^{(k)} (z))^c \cap \mathcal{B}^{(k - 1)} (z) \right).
\end{eqnarray*}
Additionally, we know from Equation~(\ref{eq:locPropNeuMS3}) that
\[
	\mathbb{P} \left( (\mathcal{B}^{(k)} (z))^c \cap \mathcal{B}^{(k - 1)} (z) \right) 
	\leq \max \left\{ 2 n e^{-z }, n \epsilon \right\}
\]
for every $k \leq k'$. Hence, we get from Equation~(\ref{eq:locPropNeuMS2}) that
\begin{eqnarray*}
	\mathbb{P} \left( \mathcal{B}^{(k_2)} (z) \vert \mathcal{B}^{(k_1)} (z) \right)
	& \geq & 1 - \frac{(k_2 - k_1) \max \left\{ 2 n e^{-z }, n \epsilon \right\}}{1 - (k_1 + 1) \max \left\{ 2 n e^{-z }, n \epsilon \right\}} \\
	& = & \frac{ 1 - (k_2 + 1) \max \left\{ 2 n e^{-z }, n \epsilon \right\}}{1 - (k_1 + 1) \max \left\{ 2 n e^{-z }, n \epsilon \right\}} 
\end{eqnarray*}
leading to the assertion.
\end{proof}

\begin{Rem}
The assumptions on the choices of~$k_1$ and~$k_2$ ensure that the lower bound in Equation~(\ref{eq:locStability}) is larger than zero and smaller than one. This lower bound for the conditional probability $\mathbb{P} \left( \mathcal{B}^{(k_2)} (z) \vert \mathcal{B}^{(k_1)} (z) \right) $ improves the lower bound of $\mathbb{P} \left( \mathcal{B}^{(k_2)} (z) \right) $ in Theorem~\ref{thm:locPropNeuMS}. However, this result allows a comparison of the established lower bounds only, but not of the exact probabilities.
\end{Rem}

\subsection{Relation to previous work}\label{sec:comparison}

In the original study by \citet{PoSp05}, the authors demonstrated propagation, separation and stability of estimates up to some constant. We will summarize these results briefly.
All associated proofs were based on the memory step. 
In this study, we have shown similar properties for the simplified algorithm, where the memory step is removed. 
However, our results are restricted to locally constant parameter functions with sharp discontinuities. 
Theoretical properties of the algorithm in case of model misspecification will be analyzed in an upcoming study.

Both studies include a certain separation property, see \citet[Section 5.5]{PoSp05} and Proposition~\ref{prop:separationMS}.
This justifies that in case of sufficiently large discontinuities smoothing is restricted to the homogeneity regions. 

For the propagation property, Polzehl and Spokoiny supposed, among other things, the statistical independence of the adaptive weights from the observations.
They then showed for $\theta(.) \equiv \theta$ that
\begin{equation}\label{eq:propPS06}
	\mathbb{P} \left( \overline{N}_i^{(k)} \mathcal{KL} \left( \hat{\theta}_i^{(k)}, \theta \right) \leq \mu \log(n) \quad \forall i \right) > 1 - 2k/n, \qquad \mu \geq 2,
\end{equation}
where~$\hat{\theta}_i^{(k)}$ denotes the adaptive estimator after modification by the memory step, see \citet[Section 3.2 and 3.3]{PoSp05}. For locally almost constant parameters they established a similar result.
Equation~(\ref{eq:propPS06}) could be improved by Proposition~\ref{prop:propCond} taking advantage of the new propagation condition introduced in Section~\ref{sec:propCond}. Setting $z := \mu \log(n)$ and $ \epsilon := c_{\epsilon} n^{-q} $ Proposition~\ref{prop:propCond} implies
\[
	\mathbb{P} \left( \overline{N}_i^{(k)} \mathcal{KL} \left( \tilde{\theta}_i^{(k)}, \theta \right) \leq \mu \log(n) \quad \forall i \right) > 1 - \max \left\{ 2/n, c_{\epsilon} / n \right\}, \qquad \mu, q \geq 2,
\]
where the additional factor~$k$ is avoided.
Theorem~\ref{thm:locPropNeuMS} sheds light on the interplay of propagation and separation during iteration. 
Here, we do not restrict the analysis to the respective homogeneous compartment as in Proposition~\ref{prop:propCond} and \citep{PoSp05}. Instead, we use the separation property to verify the propagation property for piecewise constant functions with sharp discontinuities. The resulting exponential bound in Equation~(\ref{eq:locPropNeuMS2}) complies with Equation~(\ref{eq:propPS06}) setting $z \geq q \log(n)$ and $\epsilon := c_{\epsilon} n^{-q}$ with~$c_{\epsilon} > 0$ and~$q \geq 2$.

The results on stability of estimates are difficult to compare.
Our corresponding results are stated in Propositions~\ref{prop:propCond} and~\ref{prop:locStability}. 
Polzehl and Spokoiny proved under weak assumptions stability of estimates up to some constant. More precisely, they showed that
\[
	\overline{N}_i^{(k)} \mathcal{KL} \left( \hat{\theta}_i^{(k)}, \theta_i \right) \leq \mu \log(n) 
\]
implies with probability one
\[
	\overline{N}_i^{(k)} \mathcal{KL} \left( \hat{\theta}_i^{(k^{*})}, \theta_i \right) \leq c \log(n), \qquad c := \varkappa^2 \left( \sqrt{c_1 C_{\tau}} + \sqrt{\mu} \right)^2,
\]
where~$\varkappa$ is as in Lemma~\ref{lem:A1}, $\tau := C_{\tau} \log(n)$ denotes the bandwidth of the memory kernel and $c_1 := \varkappa^2 \nu (1 - \sqrt{\nu})^{-2}$ depends on the constant~$\nu$ satisfying
$\nu_1 \leq \overline{N}_i^{(k-1)} / \overline{N}_i^{(k)} \leq \nu$
with $\nu_1, \nu \in (2/3, 1)$. Hence, the constant $c$ might be quite large. 
This result allowed to verify under smoothness conditions on the parameter function~$\theta(.)$ the optimal rate of convergence.

\section{Discussion}\label{sec:discussion}

In this section, we dwell into the propagation condition, discuss its application in practice and generalize the setting of our study.

\subsection[Independence]{(In-)dependence of the propagation condition of the parameter} \label{sec:Lambda} \label{sec:Independence} \label{sec:dependentPropCond}

The propagation condition in Definition~\ref{def:propCond} is formulated w.r.t. the unknown parameter $\theta \in \Theta$. In this section, we evaluate its dependence of this parameter. To this end, we start with a more general problem yielding a sufficient criterion.
This criterion suggests the independence of the propagation condition of the parameter~$\theta$ in case of Gaussian and exponential distribution and as a consequence of log-normal, Rayleigh, Weibull, and Pareto distribution. Additionally, we discuss the choice of~$\lambda$ if the associated function~$\mathfrak{Z}_{\lambda}$ is not independent of the paremeter~$\theta$, where we concentrate on the Poisson distribution.

We introduce a general criterion for the independence of the composition of two functions of some parameter~$\theta$.

\begin{Prop}\label{prop:zeta}
Let $f: \Omega^f \to \mathbb{R}$ and $g: \Omega^g \to \mathbb{R}$ be continuously differentiable functions with open domains $\Omega^f, \Omega^g \subseteq \mathbb{R}^2$.  We denote $\Omega^f_{\theta} := \lbrace y: (y, \theta) \in \Omega^f \rbrace$, $f_{\theta}: \Omega^f_{\theta} \to \mathbb{R}$ with $ f_{\theta}(y) := f(y,\theta) $, and analogous~$\Omega^g$ and~$g_{\theta}$. Then, we suppose $g_{\theta} (\Omega^g_{\theta}) \subseteq \Omega^f_{\theta}$ and $\left| \frac{\partial g_{\theta} }{\partial y} \right| > 0 $, such that the composition $f_{\theta} \circ g_{\theta}^{-1}: g_{\theta} (\Omega^g_{\theta}) \to \mathbb{R}$ is well-defined. 
The function 
\[
	h(z, \theta) := f_{\theta} \left( g_{\theta}^{-1}(z) \right),  \qquad (z, \theta) \in g(\Omega^g), 
\]
is independent of~$\theta$ if a variable~$\zeta(y, \theta)$ and functions~$\tilde{f}$ and~$\tilde{g}$ exist such that 
\begin{equation}\label{eq:zeta}
	\tilde{f} (\zeta) = f_{\theta} (y)
	\quad \text{ and } \quad
	\tilde{g} (\zeta) = g_{\theta} (y).
\end{equation}
\end{Prop}

\begin{proof}
Substitution with~$y := g_{\theta}^{-1} (z)$ yields $ h(g_{\theta} (y), \theta) = f \left( y, \theta \right) $ for $(y, \theta) \in \Omega^f$ and hence the total derivatives
\[
	\frac{d h}{d \theta} 
	= \frac{\partial h}{\partial z} \frac{\partial g}{\partial \theta} + \frac{\partial h}{\partial \theta} 
	= \frac{\partial f}{\partial \theta}
	\quad \text{ and } \quad
	\frac{d h}{d y} 
	= \frac{\partial h}{\partial z} \frac{\partial g}{\partial y}  
	= \frac{\partial f}{\partial y}.
\]
Then, it follows $ \frac{\partial h}{\partial z} = \frac{\partial f}{\partial y} / \frac{\partial g}{\partial y} $ and furthermore
\[
	\frac{\partial f}{\partial y} \frac{\partial g}{\partial \theta} + \frac{\partial h}{\partial \theta} \frac{\partial g}{\partial y}
	= \frac{\partial f}{\partial \theta} \frac{\partial g}{\partial y}.
\]
This leads with $\left| \frac{\partial g_{\theta} }{\partial y} \right| > 0 $ to
\[
	\frac{\partial h}{\partial \theta} 
	= \left( \frac{\partial f}{\partial \theta} \frac{\partial g}{\partial y} - \frac{\partial f}{\partial y} \frac{\partial g}{\partial \theta} \right)
	\cdot \left( \frac{\partial g }{\partial y} \right)^{-1} 
\]
such that
\[
	\frac{\partial h}{\partial \theta} = 0
	\quad \Longleftrightarrow \quad 
	\frac{\partial f}{\partial \theta} \frac{\partial g}{\partial y} = \frac{\partial f}{\partial y} \frac{\partial g}{\partial \theta}.
\]
The chain rule implies with Equation~(\ref{eq:zeta}) that indeed
\[
	\frac{\partial f}{\partial \theta} \frac{\partial g}{\partial y}
	= \frac{\partial \tilde{f}}{\partial \zeta} \frac{\partial \zeta}{\partial \theta} 
	\frac{\partial \tilde{g}}{\partial \zeta} \frac{\partial \zeta}{\partial y}
	= \frac{\partial \tilde{f}}{\partial \zeta} \frac{\partial \zeta}{\partial y} 
	\frac{\partial \tilde{g}}{\partial \zeta} \frac{\partial \zeta}{\partial \theta}
	= \frac{\partial f}{\partial y} \frac{\partial g}{\partial \theta}
\]
yielding that~$h$ is independent of~$\theta$.
\end{proof}

Now, we are well prepared to evaluate the (in-)dependence of the propagation condition in Definition~\ref{def:propCond}, and hence of the choice of~$\lambda$, of the parameter~$\theta$. 
The estimator is defined as linear combination of the terms~$T(Y_j)$, where the adaptive and the non-adaptive estimator differ only in the definition of the weights.
Thus, we approach the problem in three steps.
We start from the special case, where the estimator is restricted to a single point~$T(Y_j)$. Then, we consider the \emph{non-adaptive} estimator describing its probability density as convolution of the respective densities corresponding to the weighted observations. 
Here, we take advantage of the statistical independence of the
involved random variables~$\overline{w}_{ij}^{(k)} T(Y_j) /
\overline{N}_i^{(k)}$. In case of the \emph{adaptive} estimator we cannot
follow the same approach. This would require knowledge about the
probability distribution of the random variables~$\tilde{w}_{ij}^{(k)}
T(Y_j) / \tilde{N}_i^{(k)}$, where the adaptive weights follow an
unknown distribution. Further, these variables are not statistically
independent. To compensate the resulting lack of a theoretical proof,
we illustrate by simulations that the adaptive estimator shows almost
the same behavior as the non-adaptive estimator, if the propagation
condition is satisfied. This suggests that the probability
distribution of $\mathcal{KL} ( \tilde{\theta}_i^{(k)}, \theta )$ is
independent of~$\theta$ if the same holds true w.r.t. the non-adaptive
estimator. The single observation case is treated first.

\begin{Lem}\label{lem:zeta}
Let $\mathcal{P} = \lbrace \mathbb{P}_{\theta} \rbrace_{\theta \in \Theta}$ with $\Theta \subseteq \mathbb{R}$ be a parametric family of continuous probability distributions. 
Suppose that~$Y \sim \mathbb{P}_{\theta}$ and $T(Y) \in \Theta$ almost
surely, and that the density~$f^Y_{\theta}$ of $Y$ is continuously
differentiable.
Consider the random variable
$Z := g_{\theta}(Y) 
	:= \mathcal{KL} \left( \mathbb{P}_{T(Y(\omega))}, \mathbb{P}_{\theta} \right),$
and assume that $\frac{\partial g_{\theta} }{\partial y} \neq 0$.
The density~$f^Z_{\theta}$ of~$Z$ is independent of the parameter~$\theta$ if a variable~$\zeta(y, \theta)$ and functions~$\tilde{f}$ and~$\tilde{g}$ exist such that 
\begin{equation}\label{eq:indLambda}
	\tilde{f} (\zeta) = f^Y_{\theta}(y) \cdot \left| \frac{\partial g_{\theta} }{\partial y} (y) \right|^{-1}
	\quad \text{ and } \quad
	\tilde{g} (\zeta) = g_{\theta} (y).
\end{equation} 
\end{Lem}

\begin{proof}
The assertion follows with
\[
	h(z, \theta) := f^Z_{\theta} (z) 
	= f^Y_{\theta} \left( g_{\theta}^{-1}(z) \right) \cdot \left| \frac{\partial g_{\theta} }{\partial y} \left( g_{\theta}^{-1}(z) \right) \right|^{-1} 
\]
as special case of Proposition~\ref{prop:zeta} since $ \mathbb{P}_{\theta} \left( \left| \frac{\partial g_{\theta} }{\partial y} (y) \right| > 0 \right)= \mathbb{P}_{\theta}(T(Y) \neq \theta) =1 $.
\end{proof}

This Lemma yields the desired results for Gaussian and Gamma-distributed observations 
.

\begin{Ex}\label{ex:Independence}
We consider the same setting as in Lemma~\ref{lem:zeta}. In the following cases, the density of~$Z$ is independent of the parameter~$\theta$.
\begin{itemize}
\item $\mathcal{P} = \left\{ \mathcal{N} (\theta, \sigma^2) \right\}_{\theta \in \Theta}$ with~$\sigma > 0$ fixed:
Equation~(\ref{eq:KL}) and Table~\ref{tab:contDistr} yield for the Kullback-Leibler divergence of $\mathbb{P}_{\theta}, \mathbb{P}_{\theta'} \in \mathcal{P}$ the explicit formula
\[
	\mathcal{KL} \left( \theta, \theta' \right) =
	\frac{(\theta - \theta')^2 }{ 2 \sigma^2 }
	\quad \text{ such that } \quad
	\frac{\partial g_{\theta} }{\partial y} (y) = \frac{y - \theta}{\sigma^2}.
\]
Since $	f^Y_{\theta} (y) = \exp \left(- \frac{(y-\theta)^2}{2 \sigma^2} \right)/ \sqrt{2 \pi \sigma^2 }$
we get the independence of~$\theta$ from Lemma~\ref{lem:zeta} by setting
\[
	\zeta(y, \theta) := y - \theta,
	\qquad
	\tilde{f} (\zeta) := \frac{\sigma \, e^{-\frac{\zeta^2}{2 \sigma^2}}}{\zeta \sqrt{2 \pi}} ,
	\quad \text{ and } \quad
	\tilde{g} (\zeta) := \frac{\zeta^2}{2 \sigma^2}.
\]
\item $ \mathcal{P} = \left\{ \Gamma (p, \theta) \right\}_{\theta \in \Theta}$ with~$p >0$ fixed: 
It holds $f^Y_{\theta} (y) = \frac{y^{p-1} \, e^{-y/ \theta}}{\theta^p \Gamma(p)}$, such that
\[
	\mathcal{KL} \left( \theta, \theta' \right)
	= p \left[ \theta/\theta' - 1 - \ln \left( \theta/\theta' \right) \right]
	\quad \text{ and } \quad
	\frac{\partial g_{\theta} }{\partial y} (y) = p \left( \frac{1}{\theta} - \frac{1}{y} \right).
\]
Thus, Lemma~\ref{lem:zeta} can be applied with
\[
	\zeta(y, \theta) := \frac{ y }{ \theta },
	\qquad
	\tilde{f} (\zeta) = \frac{\zeta^p \, e^{- \zeta}}{p(\zeta - 1) \Gamma(p)}
	\quad \text{ and } \quad
	\tilde{g} (\zeta) = p \left[ \zeta - 1 - \ln \zeta \right].
\]
\end{itemize}
\end{Ex}

This extends to non-adaptive linear combinations as follows.
Lemma~\ref{lem:zeta} can be applied w.r.t. the non-adaptive estimator with~$Y := \overline{\theta}_i^{(k)}$ considering the composition of the density~$f_{\theta}^{\overline{\theta}_i^{(k)}}$ and the Kullback-Leibler divergence described by the function~$g_{\theta}$. 
While the latter depends on the assumed parametric family~$\mathcal{P}$ only, the density~$f_{\theta}^{\overline{\theta}_i^{(k)}}$ is determined via convolution of the probability densities of~$\overline{w}_{ij}^{(k)} T(Y_j) / \overline{N}_i^{(k)}$, where $Y_j \sim \mathbb{P}_{\theta} \in \mathcal{P}$. Hence, it depends directly on the function~$T(.)$ introduced in Assumption~(\ref{A1}).

\begin{Thm}\label{thm:Independence}
Let $\mathcal{P} = \lbrace \mathbb{P}_{\theta} \rbrace_{\theta \in \Theta}$ with $\Theta \subseteq \mathbb{R}$ be a parametric family of probability distributions.
We consider the random variable 
\[
	Z := g_{\theta}(\overline{\theta}_i^{(k)}) := \left[ \omega \mapsto \mathcal{KL} \left( \mathbb{P}_{\overline{\theta}_i^{(k)} (\omega)}, \mathbb{P}_{\theta} \right) \right], 
\]
where~$\overline{\theta}_i^{(k)}$ denotes the non-adaptive estimator
depending on the observations $Y_j \overset{\mathrm{iid}}{\sim}
\mathbb{P}_{\theta}$ with $j \in \lbrace 1,...,n \rbrace$ and
some~$\theta \in \Theta$. The density of~$Z$ is independent of the
parameter~$\theta$ in the following cases.
\begin{itemize}
\item $ \mathcal{P} = \left\{ \mathcal{N} (\theta, \sigma^2) \right\}_{\theta \in \Theta}$ with~$\sigma > 0$ fixed;
\item $ \mathcal{P} = \left\{ \mathrm{log}\mathcal{N} (\theta, \sigma^2) \right\}_{\theta \in \Theta}$ with~$\sigma > 0$ fixed;
\item $ \mathcal{P} = \left\{ \mathrm{Exp} (1/\theta) \right\}_{\theta \in \Theta}$;
\item $ \mathcal{P} = \left\{ \mathrm{Rayleigh} (\theta) \right\}_{\theta \in \Theta}$;
\item $ \mathcal{P} = \left\{ \mathrm{Weibull} (\theta, k) \right\}_{\theta \in \Theta}$ with~$k > 0$;
\item $ \mathcal{P} = \left\{ \mathrm{Pareto} (x_m, \theta) \right\}_{\theta \in \Theta}$ with~$x_m \geq 1$.
\end{itemize}
\end{Thm}

\begin{proof}
The non-adaptive estimator is defined as weighted mean of~$T(Y_j)$ with $j = 1,..,n$. We get from Table~\ref{tab:contDistr} that
\begin{itemize}
\item $T(Y) = \ln(Y) \sim \mathcal{N} (\mu, \sigma^2)$ if $Y \sim \mathrm{log} \mathcal{N} (\mu, \sigma^2)$; 
\item $T(Y) = Y^2 \sim \mathrm{Exp} \left( \frac{1}{2 \theta^2} \right)$ if $ Y \sim \mathrm{Rayleigh} (\theta)$;
\item $T(Y) = Y^k\sim \mathrm{Exp} \left( \frac{1}{\theta^k} \right)$ if $ Y \sim \mathrm{Weibull} (\theta, k)$ with~$k > 0$;
\item $T(Y) = \ln \left( y / x_m \right) \sim \mathrm{Exp} \left( \theta \right)$ if $ Y \sim \mathrm{Pareto} (x_m, \theta)$.
\end{itemize}
Hence, in each of these cases, the non-adaptive estimator follows the same distribution as for Gaussian or exponentially distributed observations. 
Additionally, the corresponding Kullback-Leibler divergences coincide with the respective divergences of Gaussian or exponential distributions. Therefore, it suffices to consider Gaussian and exponential distribution.

In the Gaussian case, it follows from the statistical independence of the observations~$ Y_j \overset{\mathrm{iid}}{\sim} \mathcal{N} (\theta, \sigma^2) $, that
\[
	\overline{\theta}_i^{(k)} \sim \mathcal{N} \left( \theta, \sigma^2 C_i\right), \qquad \text{where } C_i := \sum_{j=1}^n \left( \overline{w}_{ij}^{(k)} / \overline{N}_i^{(k)} \right)^2.
\]
Hence, the non-adaptive estimator is again Gaussian distributed and the independence of~$\theta$ follows analogous to Example~\ref{ex:Independence}, where~$\zeta$ and~$\tilde{g}$ remain unchanged and
\[
	\tilde{f} (\zeta) := \frac{\sigma^2 }{ \zeta \sigma_i \sqrt{2 \pi}} \exp \left( - \frac{\zeta^2}{2 \sigma_i^2} \right).
\]
Next, we consider the exponential distribution supposing $Y_j \overset{\mathrm{iid}}{\sim} \mathrm{Exp} ( 1/ \theta ) $. 
We distinguish two cases. 
First, if all non-zero weights are equal, and hence
$\overline{w}_{ij}^{(k)} \in \lbrace 0, 1 \rbrace$ as
$\overline{w}_{ii}^{(k)} = 1$ for all~$k$, then  the non-adaptive estimator~$\overline{\theta}_i^{(k)}$ is Gamma-distributed, i.e. 
\[
	\overline{\theta}_i^{(k)} \sim \Gamma \left( \overline{N}_i^{(k)}, \theta / \overline{N}_i^{(k)} \right).
\] 
This yields the desired independency of~$\theta$ via Example~\ref{ex:Independence} setting~$Y := \overline{\theta}_i^{(k)}$. 
Next, in the general case, we require the existence of non-zero weights $ \overline{w}_{ij}^{(k)} \neq \overline{w}_{ij'}^{(k)} $ with $j,j' \in \lbrace 1,...,n \rbrace$. If $Y_j \sim \mathrm{Exp}(1/\theta)$ then it holds $a_j Y_j \sim \mathrm{Exp} ( 1/(\theta a_j) )$ for all~$a_j > 0$, where we denote $a_j := \overline{w}_{ij}^{(k)}/ \overline{N}_i^{(k)}$ for the sake of simplicity.
The linear combination $Y := a_1 Y_1 + a_2 Y_2$ with $a_1 \neq a_2$ has the density
\begin{eqnarray*}
	f^Y (y) &=& \left( f^{a_1 Y_1} \ast f^{a_2 Y_2} \right) (y) \\
	&=& \int_0^y \frac{1}{\theta a_1} e^{- \frac{y-z}{\theta a_1}} \frac{1}{\theta a_2} e^{- \frac{z}{\theta a_2}} dz \\
	&=& \frac{ e^{- \frac{y}{\theta a_1}} }{\theta^2 a_1 a_2} \int_0^y e^{- z \frac{a_1 - a_2}{\theta a_1 a_2}} dz \\
	&=& \frac{ e^{- \frac{y}{\theta a_1}}}{\theta^2 a_1 a_2}  \cdot \frac{\theta a_1 a_2}{a_2 - a_1} \left( e^{- y \frac{a_1 - a_2}{\theta a_1 a_2}} -1 \right) \\
	&=& \frac{ 1 }{\theta (a_1 - a_2)} e^{- \frac{y}{\theta a_1}} 
	- \frac{ 1 }{\theta (a_1 - a_2)} e^{- \frac{y}{\theta a_2}}\\
	&=& \frac{ a_1 }{a_1 - a_2} f^{a_1 Y_1} (y) - \frac{ a_2 }{a_1 - a_2} f^{a_2 Y_2} (y),
\end{eqnarray*}
which is a weighted sum of the component densities. Therefore, this extends to the more general case $\overline{Y} := a_1 Y_1 + ... + a_m Y_m$ with $a_j \neq a_{j'}$ for all~$j \neq j'$. 
Including the case of equal weights $ a_j = a_{j'} $ for some $j,j' \in \lbrace 1,...,n \rbrace$ we conclude that
\[
	f_{\theta}^{\overline{\theta}_i^{(k)}} = \sum_{j=1}^m c_j f_j,\]
where the constants~$c_j \in \mathbb{R}$ depend again on $a_1,...,a_m$ only. The densities~$f_j$ follow the distribution~$\Gamma(m_j, \theta a_j )$, where ~$m_j$ denotes the number of observations~$Y_{j'}$ with weights $a_{j'} = a_j$.
Thus, we get from Example~\ref{ex:Independence} the independence
of~$\theta$ for each summand~$c_j f_j$ yielding the assertion for weighted sums of exponentials. 
\end{proof}

\begin{Rem}
We know from Example~\ref{ex:Independence} that the random variable $ \left[ \omega \mapsto \mathcal{KL} \left( \mathbb{P}_{T(Y(\omega))}, \mathbb{P}_{\theta} \right) \right] $
is independent of the parameter~$\theta$ if the observations follow
a Gamma distribution. 
However, the probability distribution of the corresponding non-adaptive estimator has a quite sophisticated form \citep{MR695077, MR818052}, where the corresponding summands could not been proven to be independent of~$\theta$. Though, in case of a location kernel that attains only values in~$\lbrace 0,1 \rbrace$ we get
\[
	Y_j \overset{\mathrm{iid}}{\sim} \Gamma (p, \theta) \, \Longrightarrow \, \overline{\theta}_i^{(k)} \sim \Gamma (\overline{N}_i^{(k)} p, \theta/ \overline{N}_i^{(k)}) \qquad \text{ if } \overline{w}_{ij}^{(k)} \in \lbrace 0, 1 \rbrace \text{ for all } j.
\]
This yields via Example~\ref{ex:Independence} the independence of~$\theta$.
The same holds true for the Erlang and scaled chi-squared distribution since 
\[
	\mathrm{Erlang}(n, 1/ \theta) = \Gamma(n, \theta)
	\quad \text{ and } \quad
	Y \sim \Gamma(k/2,2 \theta/ k) \text{ if } k Y/ \theta \sim \chi^2 (k) = \Gamma(k/2,2).
\]
\end{Rem}

The new propagation condition is included into the R-package
\texttt{aws} \citep{aws}. First tests yield smaller values of the
adaptation bandwidth~$\lambda$ than the previous version of the
propagation condition, hence allowing for better smoothing results with a smaller estimation bias. 

In Figures~\ref{fig:Gauss} and~\ref{fig:Exp}, we show some examples to illustrate the close relation of the
adaptive and the non-adaptive estimator under a satisfied propagation
condition. 
Both Theorem~\ref{thm:Independence} and the numerical simulations suggest the independence of the propagation condition of the parameter~$\theta$.

The plots have been realized using the function \texttt{awstestprop} on a two-dimensional design with $5000 \times 5000$ points and the same kernels as in Equation~(\ref{eq:kernels}).
The maximal location bandwidth~$h^{(k^*)}$ was set to~$50$ requiring~$38$ iteration steps. 
Running the simulation with different parameters~$\theta$ yield exactly the same plots. In Figure~\ref{fig:Gauss}, we show the results for the Gaussian distribution with three different values of~$\lambda$. In Figure~\ref{fig:Exp}, we consider the same setting w.r.t. the exponential distribution.

\begin{figure}
 \includegraphics[width = 0.32\textwidth]{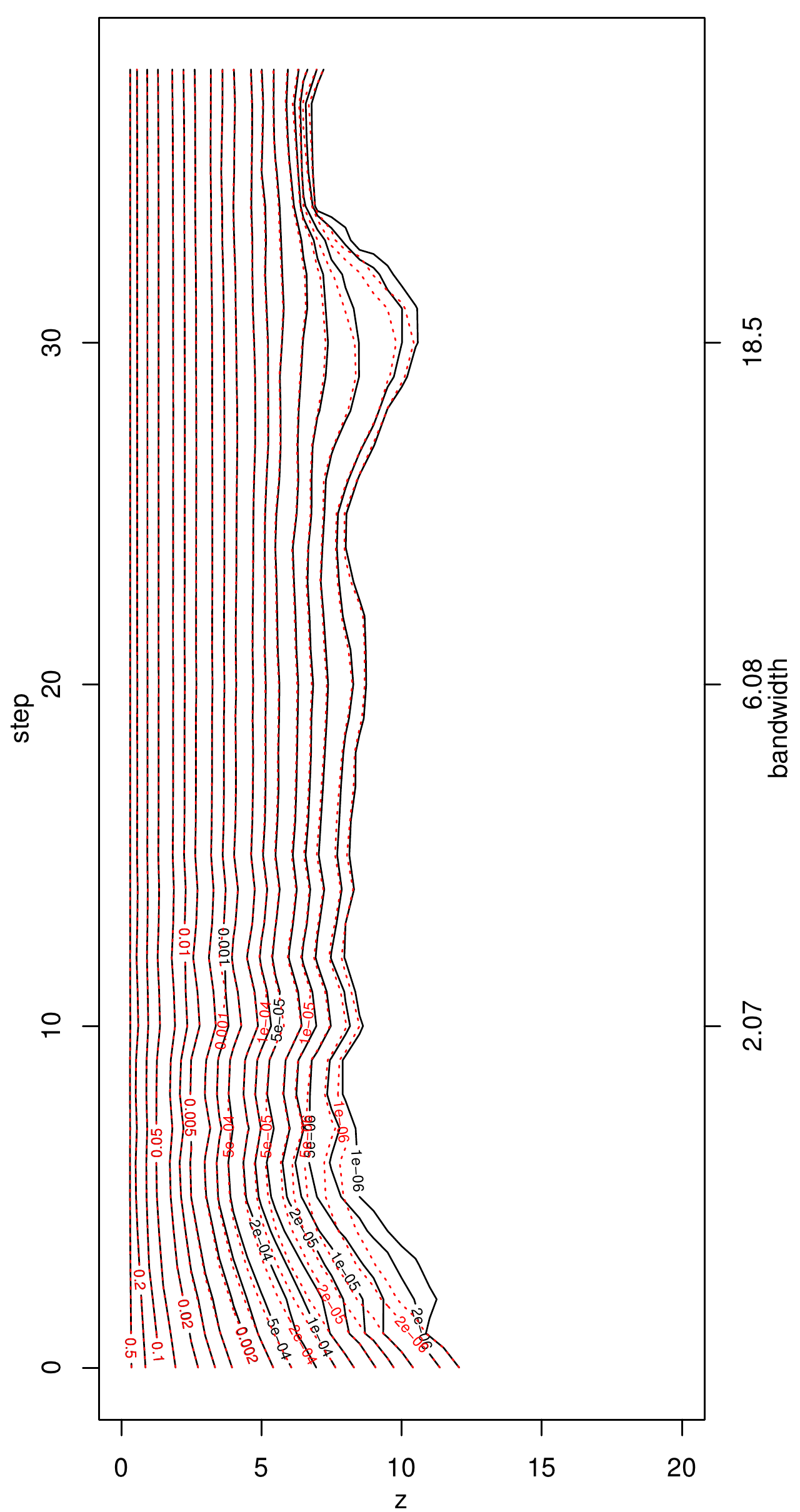} \hspace{1 pt}
 \includegraphics[width = 0.32\textwidth]{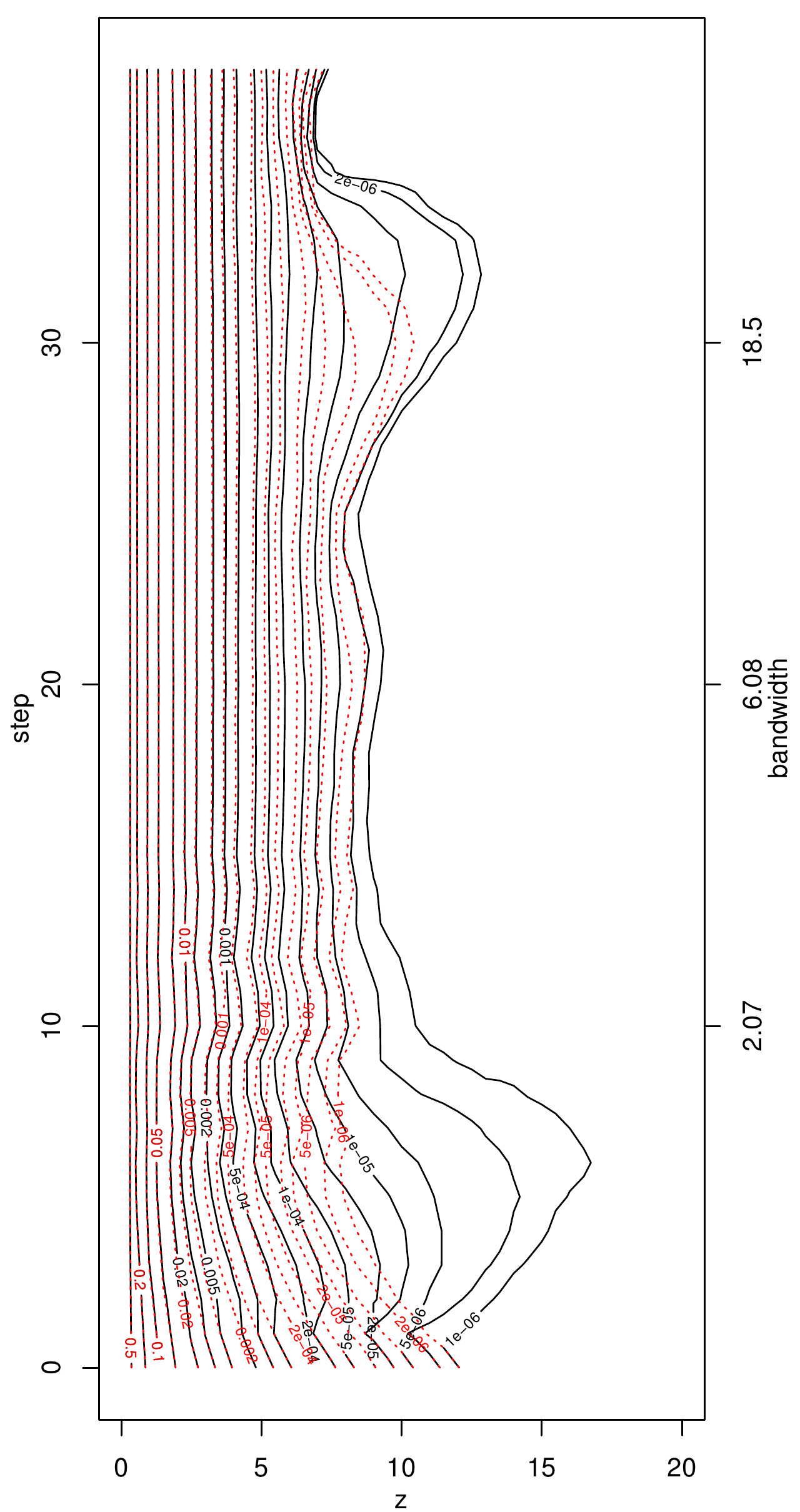} \hspace{1 pt}
 \includegraphics[width = 0.32\textwidth]{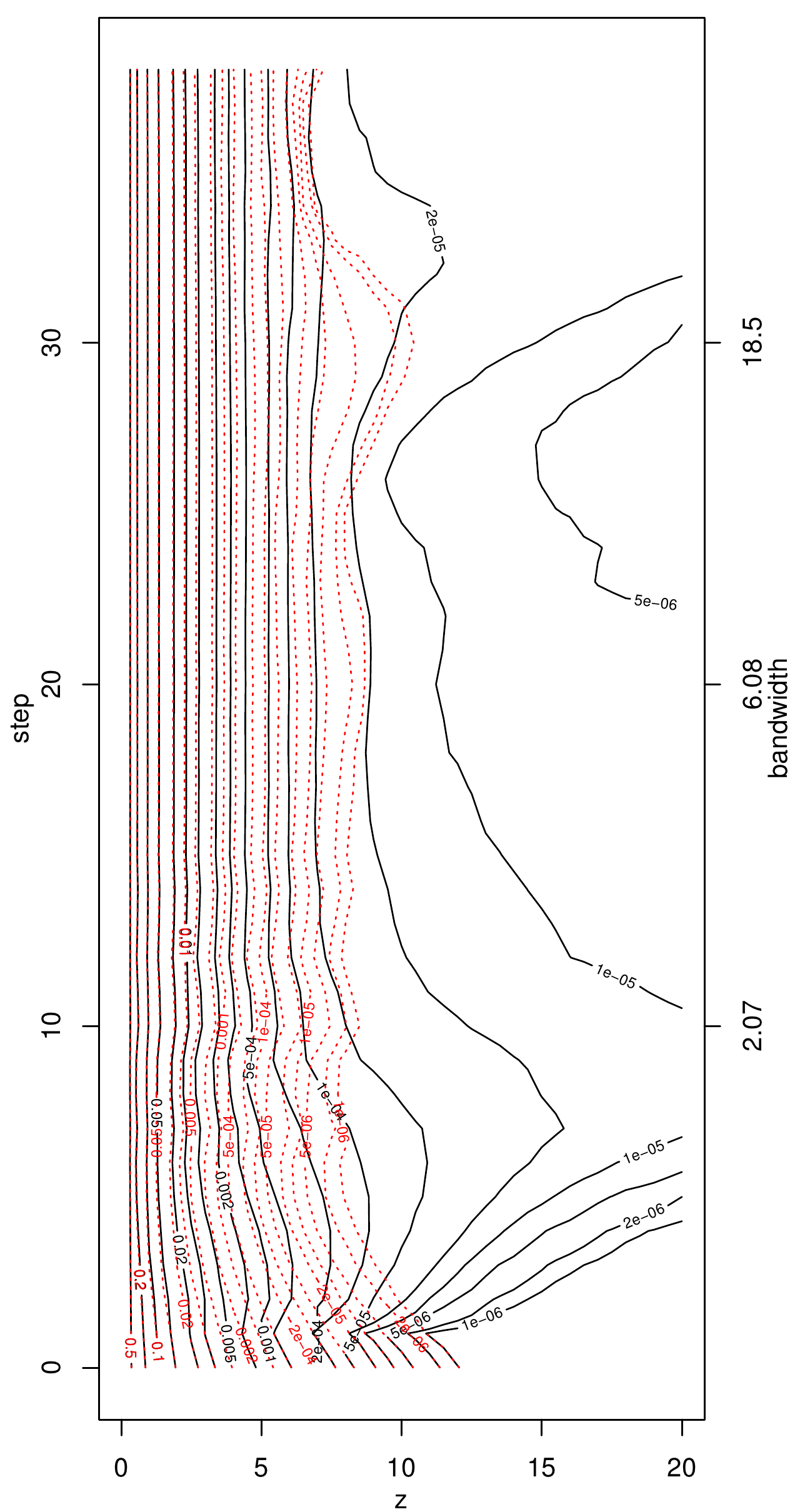}
 \caption{{\footnotesize Plots of the propagation condition for the Gaussian distribution with (f.l.t.r.) $\lambda = 22.4, 13.6, 9.72$. The isolines of the probability~$p$ for values between~$10^{-6}$ and~$0.5$ are plotted w.r.t. the location bandwidth~$h^{(k)}$ described by the iteration step~$k$ and the corresponding value $z = \mathfrak{Z}_{\lambda} (k, p; \theta = 1)$. The black solid lines represent the isolines of the adaptive estimator, the red dotted lines correspond to the non-adaptive estimator. }}
 \label{fig:Gauss}
\end{figure}

\begin{figure}
 \includegraphics[width = 0.32\textwidth]{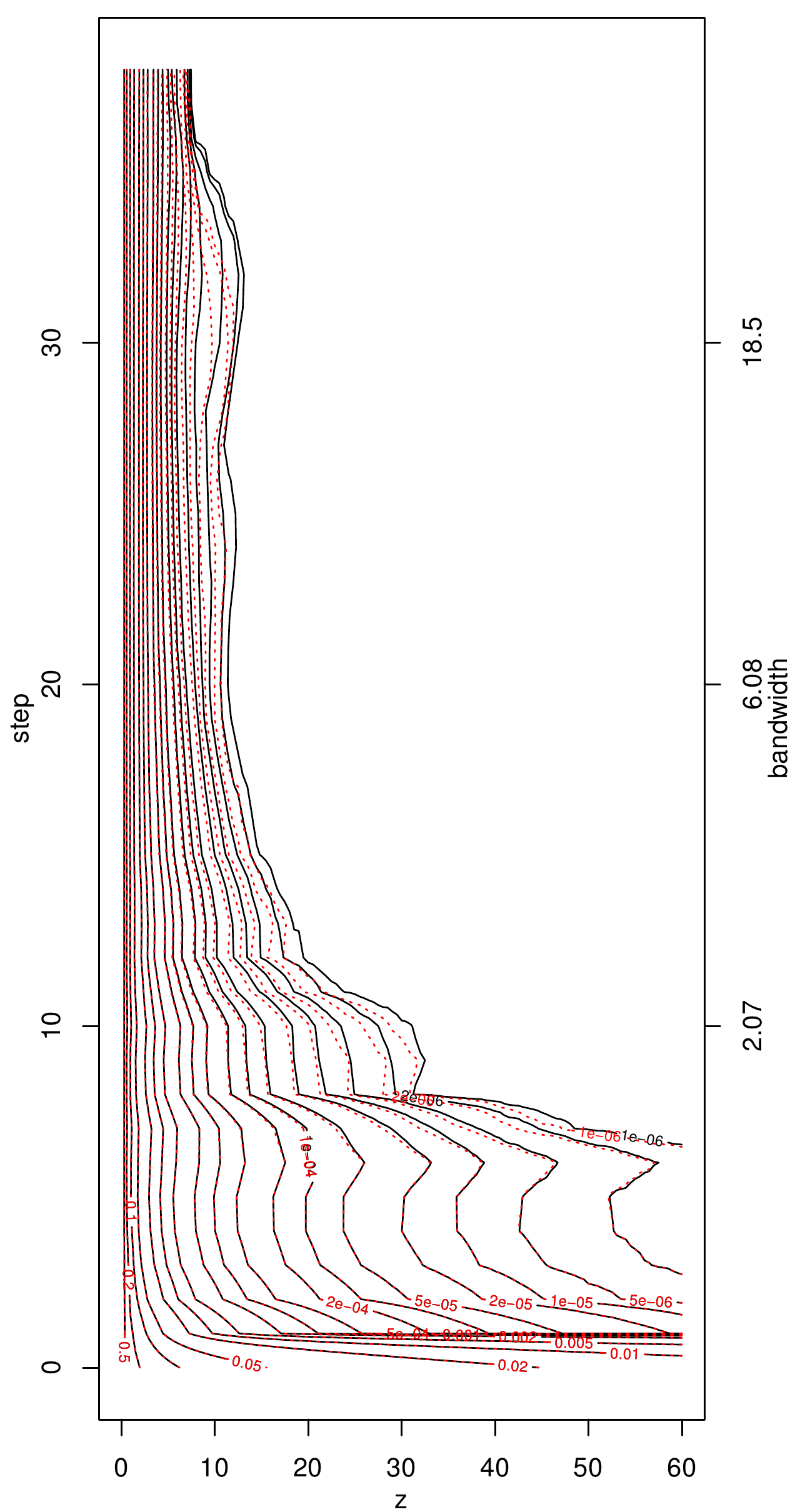} \hspace{1 pt}
 \includegraphics[width = 0.32\textwidth]{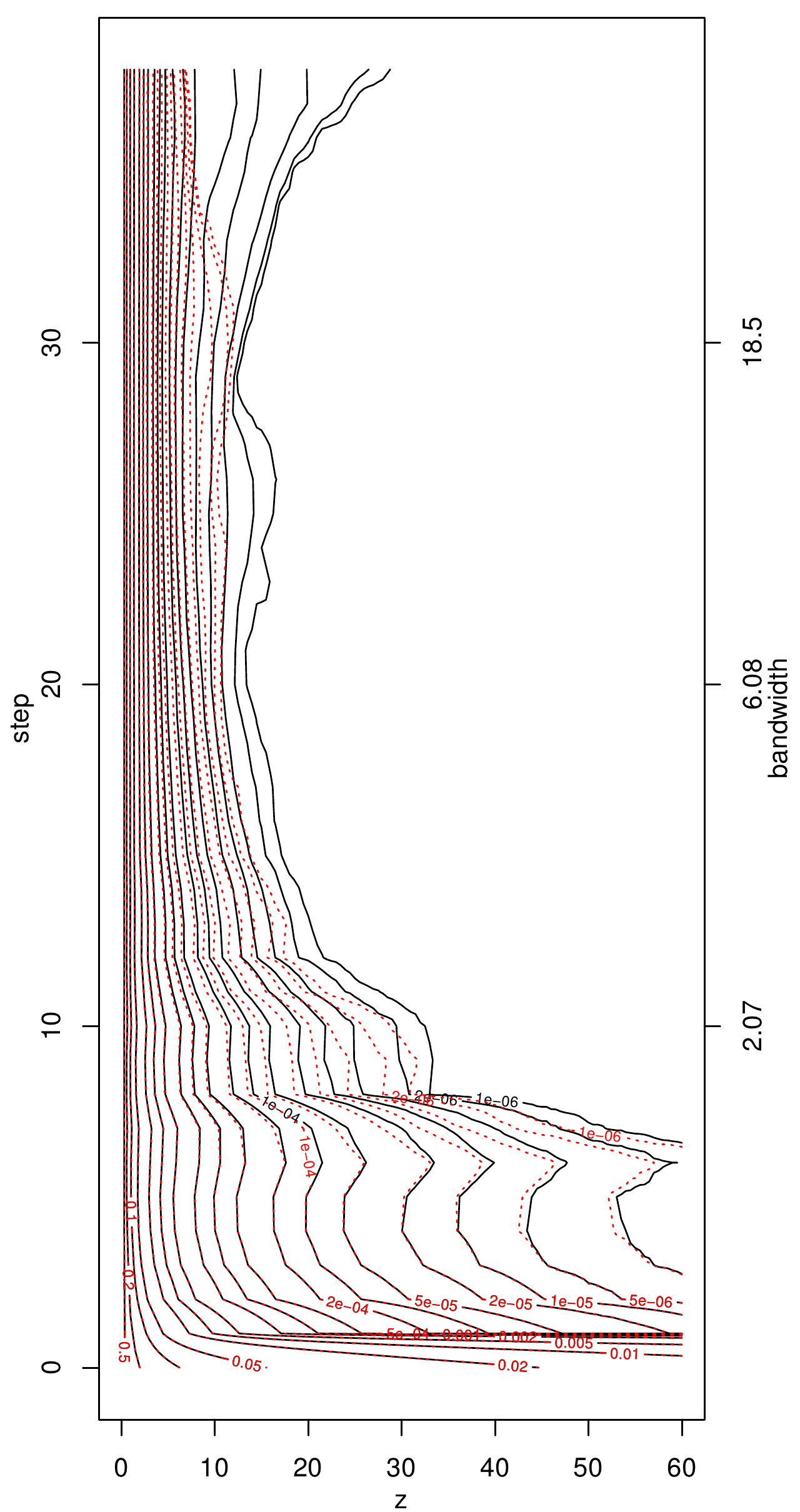} \hspace{1 pt}
 \includegraphics[width = 0.32\textwidth]{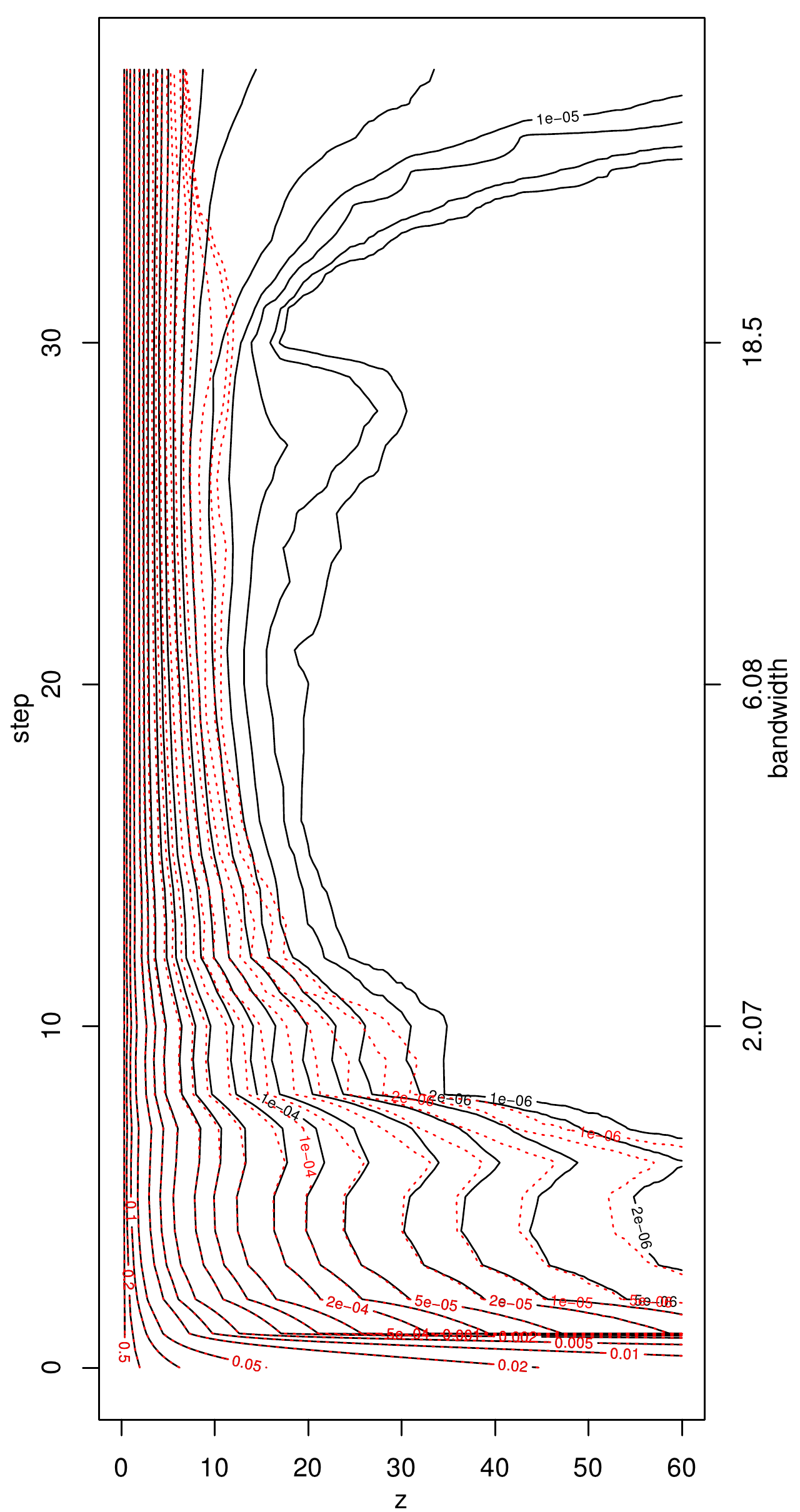}
 \caption{{\footnotesize Plots of the propagation condition for the exponential distribution with (f.l.t.r.) $\lambda = 13.2, 10.2, 8.78$.}}
 \label{fig:Exp}
\end{figure}


Finally, we discuss how to proceed if the function~$ \mathfrak{Z}_{\lambda} $ depends on the parameter~$\theta$. 
We want to ensure that our choice of the adaptation bandwidth~$\lambda$ is in accordance with the propagation condition for all~$\theta_i$, $i \in \lbrace 1,...,n \rbrace$. Certainly, we do not know the exact parameters~$\lbrace \theta_i \rbrace_i$. Instead, we could analyze the monotonicity of the optimal choice~$\lambda_{\mathrm{opt}} (\epsilon, \theta)$, see Remark~\ref{rem:propCond}, for a fixed constant~$\epsilon>0$ and varying parameters~$\theta \in \Theta$. 
For the sake of simplicity, we prefer to observe for a fixed adaptation bandwidth~$\lambda$ and varying parameters~$\theta$ for which probabilities~$p$ the propagation condition is satisfied. This can be done by the function \texttt{awstestprop} in the R-package \texttt{aws}. Thus, we get for every~$\theta$ the corresponding value~$\epsilon_{\lambda}(\theta)$. Then, $\epsilon_{\lambda} (\theta) \geq \epsilon_{\lambda} (\theta')$ indicates that the parameter~$\theta$ requires a larger adaptation bandwidth than the parameter~$\theta'$. 
Taking the range of our observations into account, 
we tempt to identify a finite number of parameters $\theta^* \in \Theta$ such that every~$\lambda$ that satisfies the propagation condition 
for these parameters $\theta^* \in \Theta$ remains valid with high probability for the unknown parameters~$\theta_i$, $i \in \lbrace 1,...,n \rbrace$.

For  observations following a Poisson distribution it turned out that different parameters~$\theta$ yield comparable propagation levels~$\epsilon_{\lambda} (\theta)$, even though the resulting isolines differ clearly. 
This is illustrated in Figure~\ref{fig:Poisson}, where we consider the same kernels as in Equation~\eqref{eq:kernels}, a regular design with $5000 \times 5000$ points, and $h^{(k^*)} = 50$, i.e.~$38$ iteration steps. In case of Bernoulli distributed observations it seems to be recommendable to ensure the propagation condition for~$\theta^* := 0.5$.
In both cases the implemented algorithm avoids that the
Kullback-Leibler divergence becomes infinity by slightly shifting the estimator.

\begin{figure}[p]
 \includegraphics[width = 0.245\textwidth]{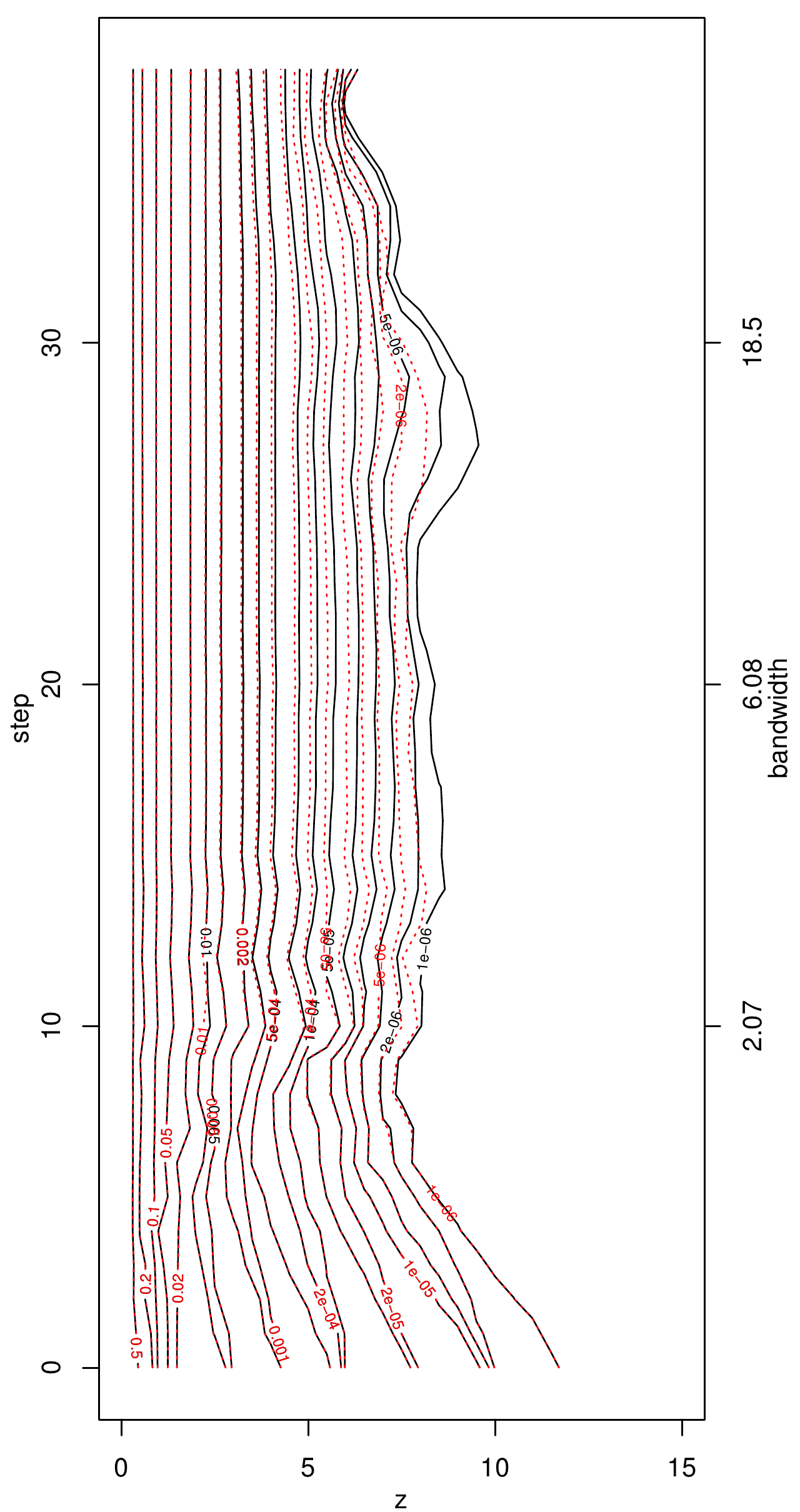}
 \includegraphics[width = 0.245\textwidth]{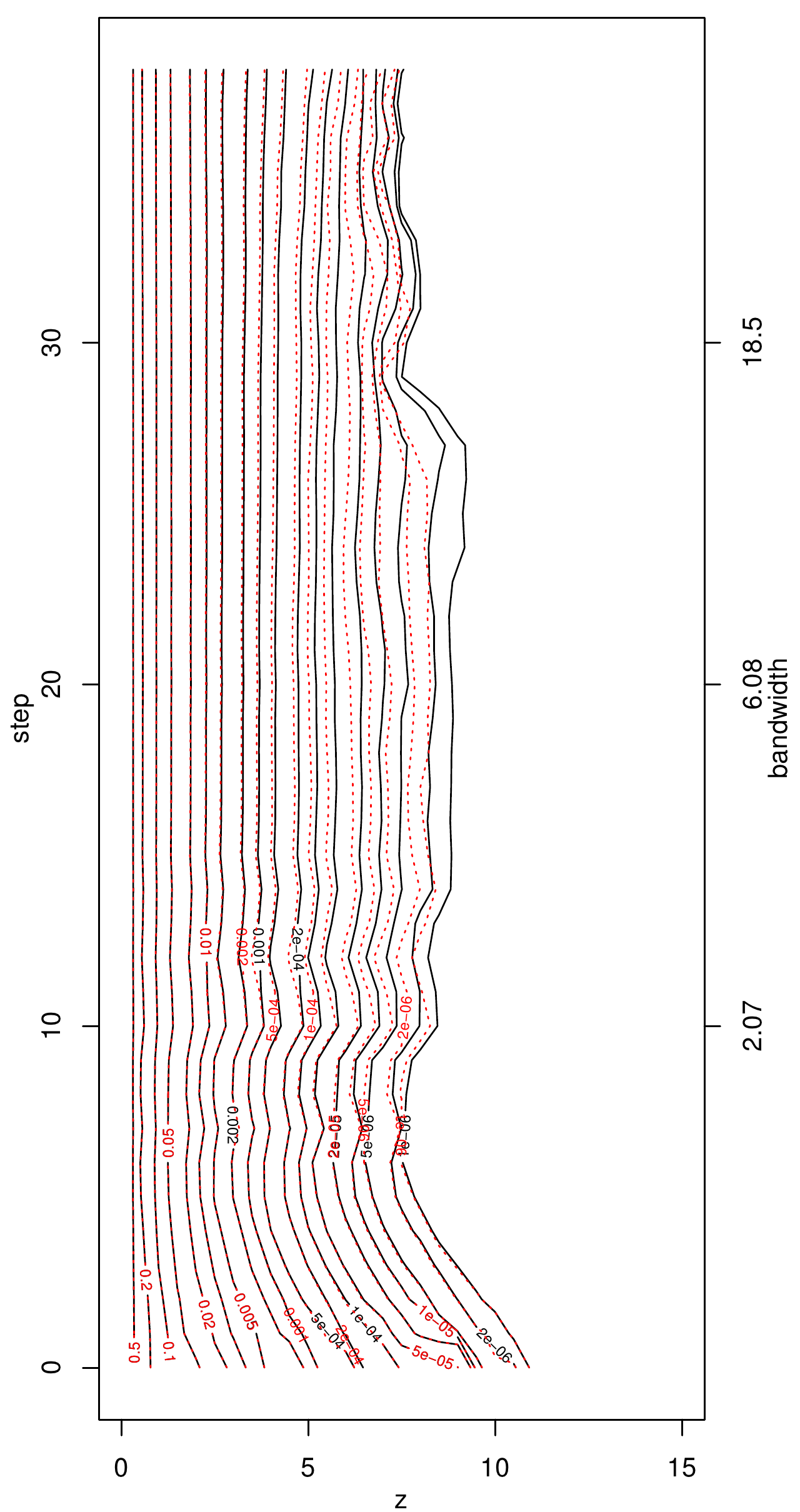}
 \includegraphics[width = 0.245\textwidth]{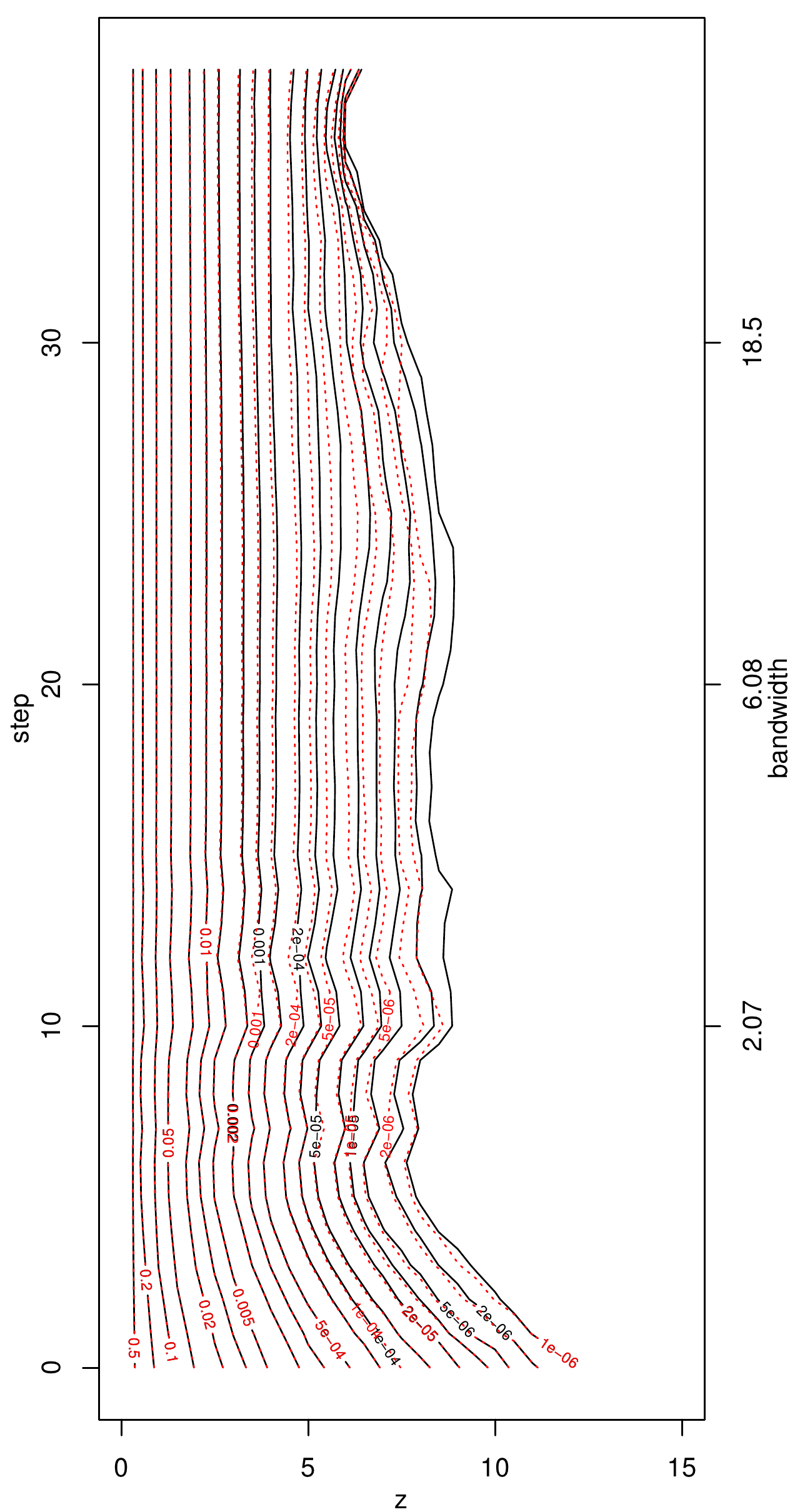}
 \includegraphics[width = 0.245\textwidth]{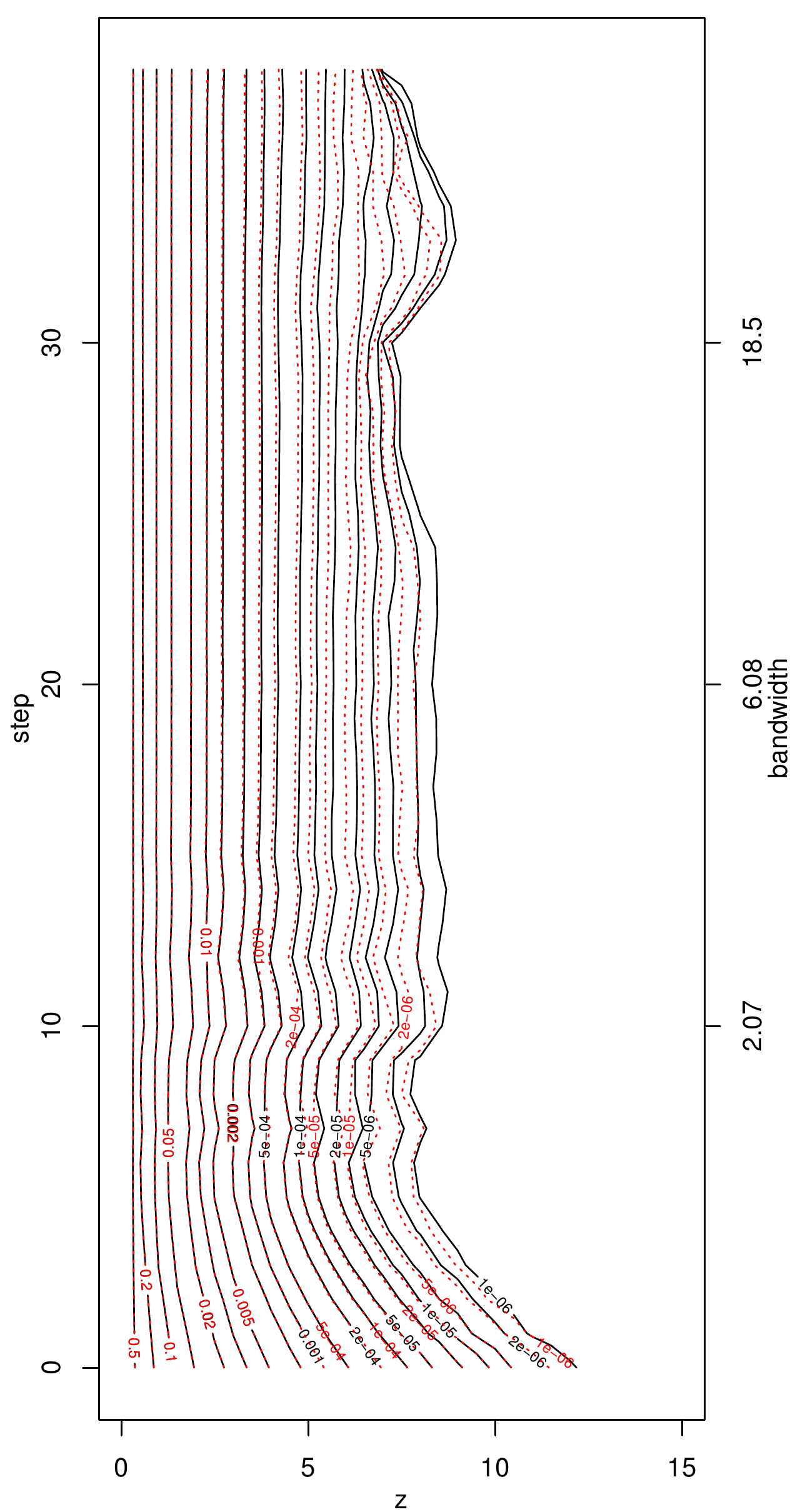}\\
 \includegraphics[width = 0.245\textwidth]{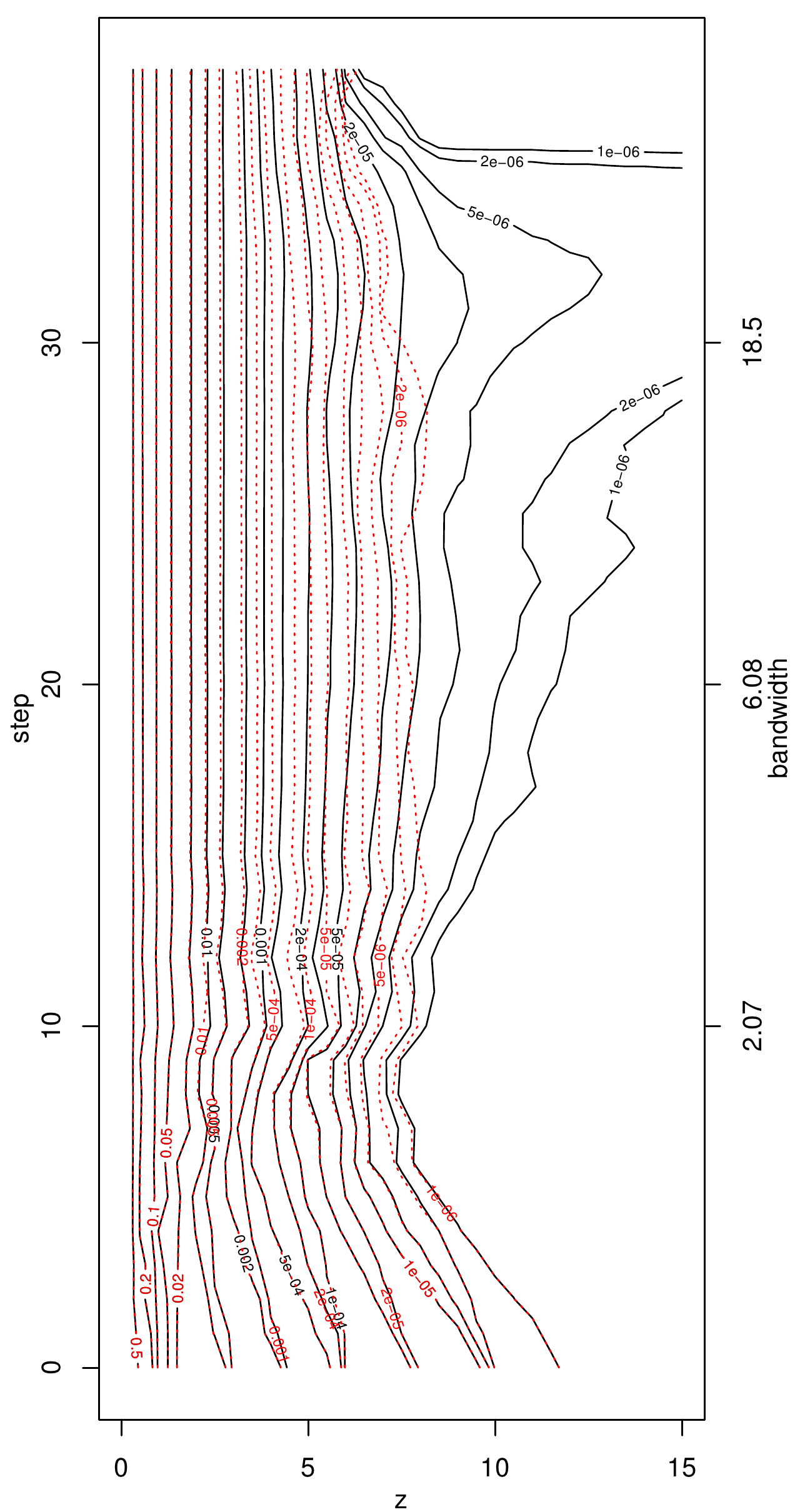}
 \includegraphics[width = 0.245\textwidth]{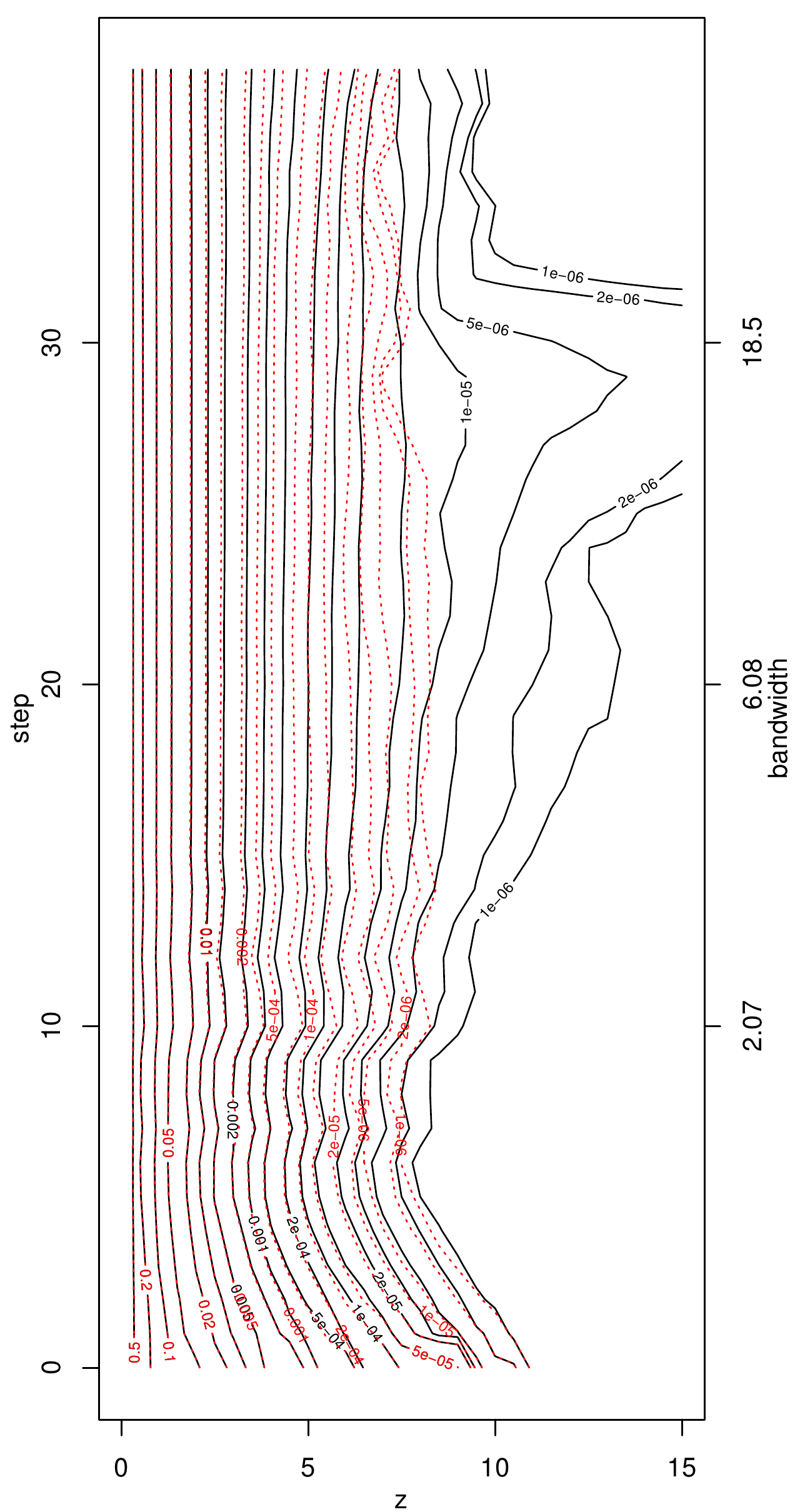}
 \includegraphics[width = 0.245\textwidth]{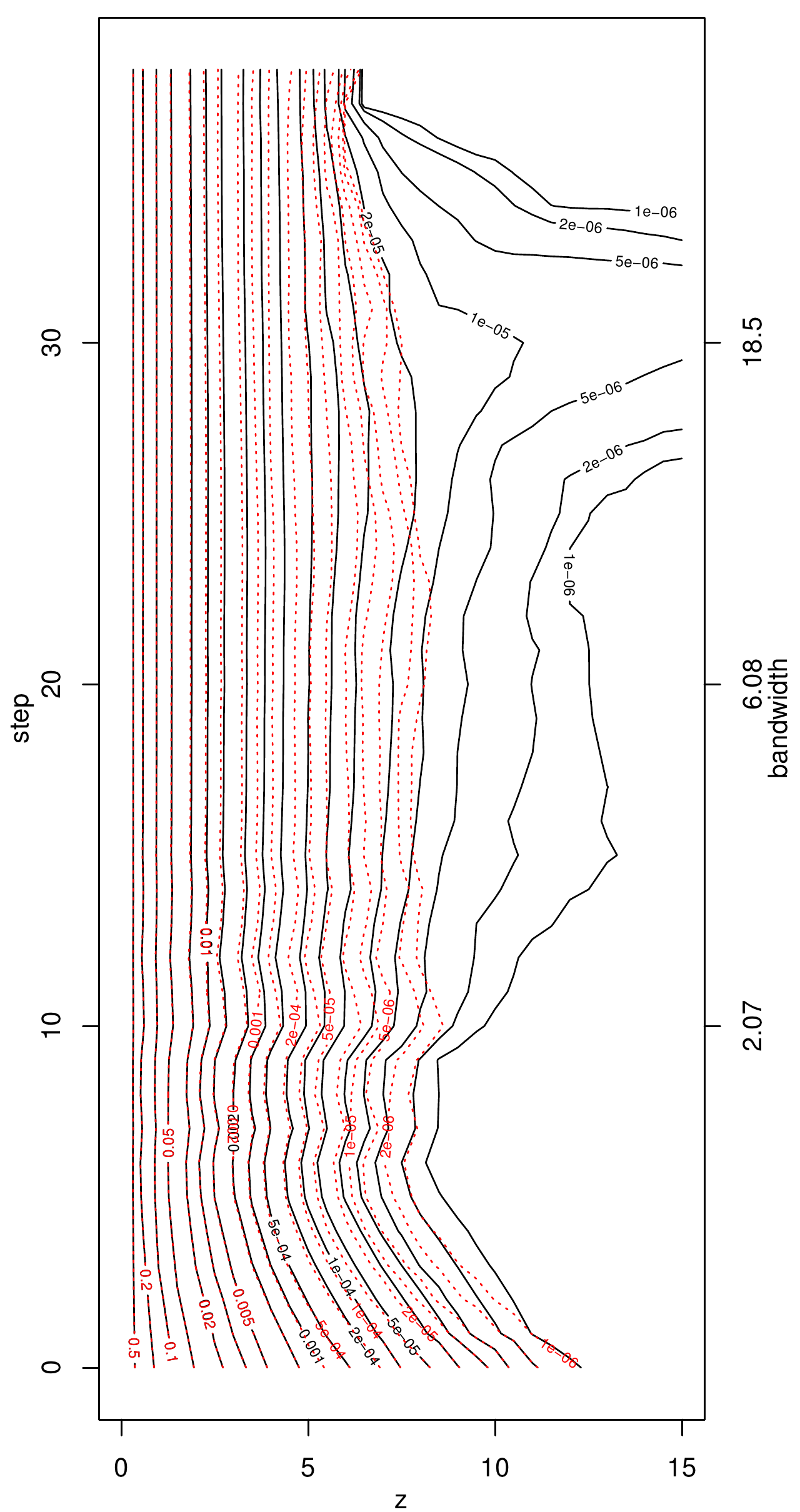}
 \includegraphics[width = 0.245\textwidth]{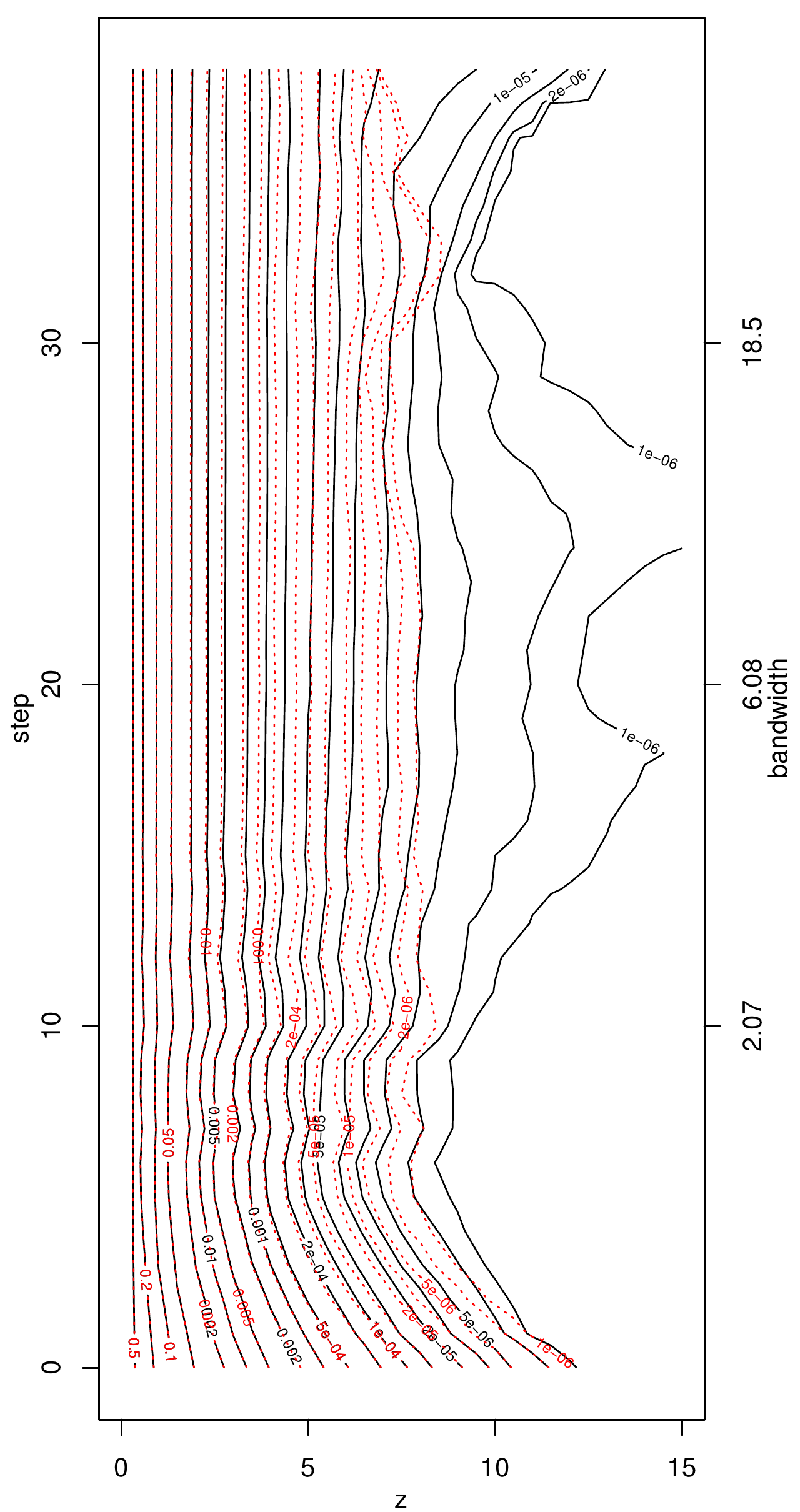}\\
 \includegraphics[width = 0.245\textwidth]{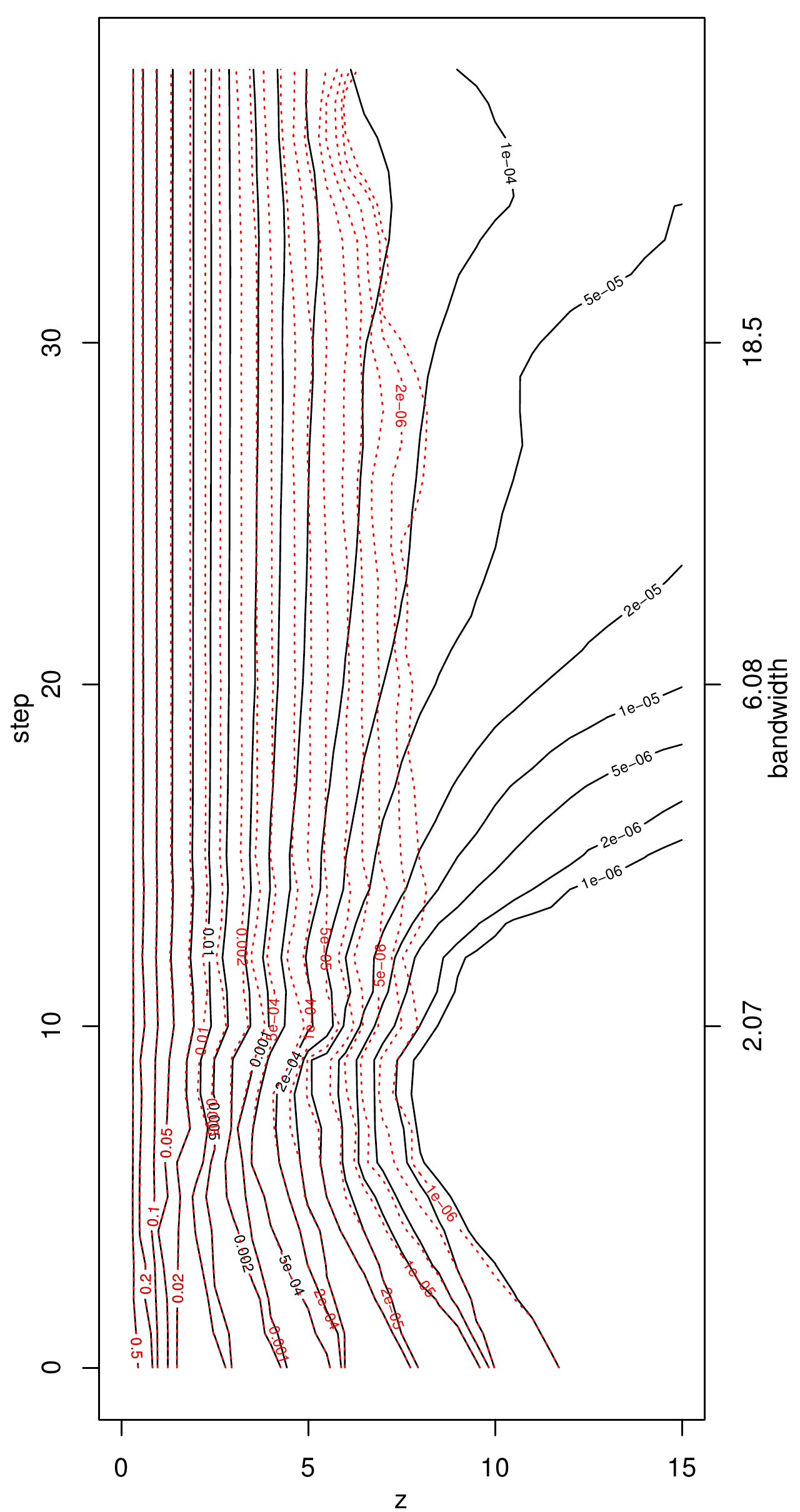}
 \includegraphics[width = 0.245\textwidth]{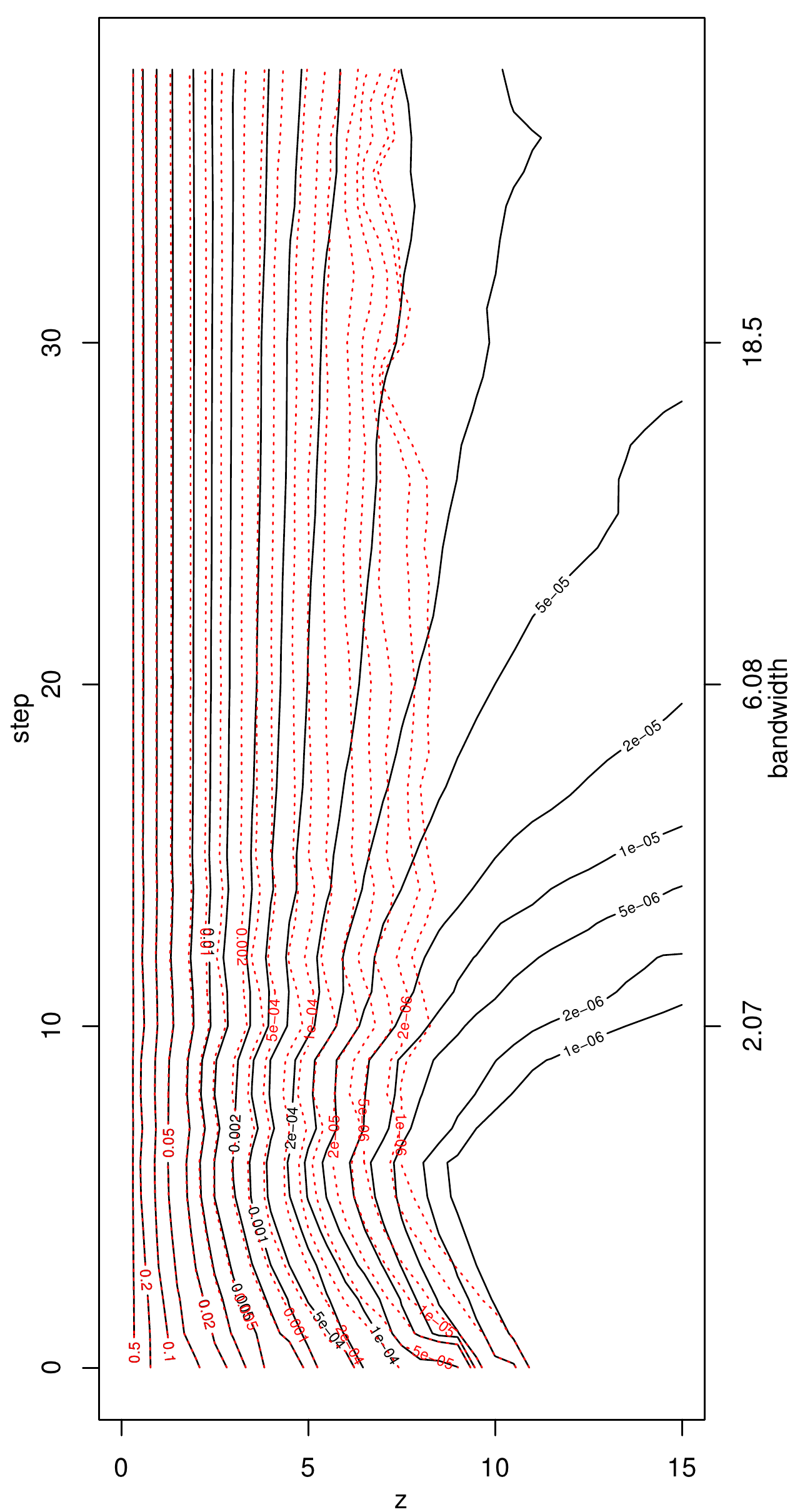}
 \includegraphics[width = 0.245\textwidth]{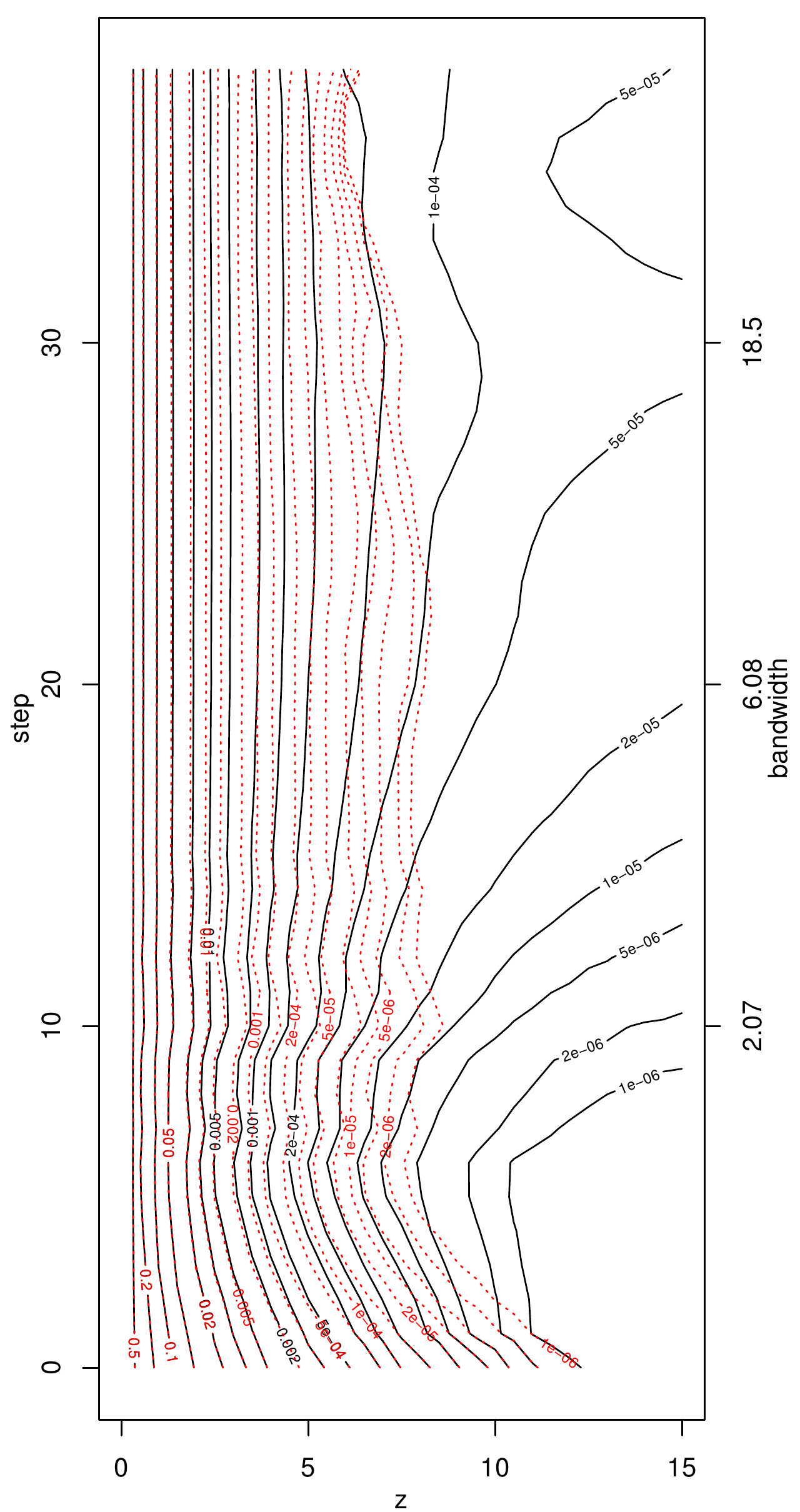}
 \includegraphics[width = 0.245\textwidth]{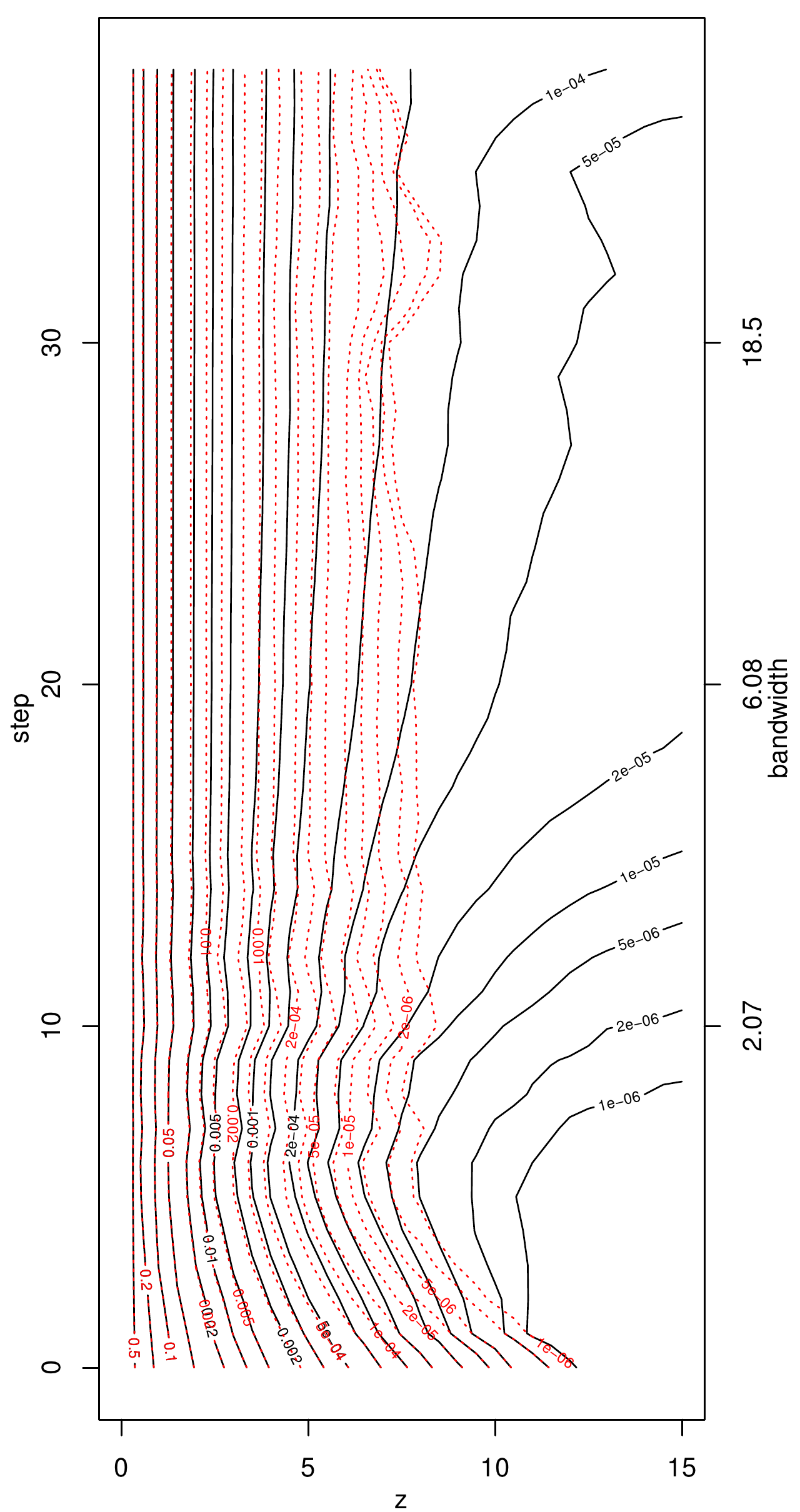}
 \caption{{\footnotesize Plots of the propagation condition for the Poisson distribution with (f.l.t.r.) $\theta = 1, 10, 100, 1000$ and (from top to bottom) $\lambda = 13.2, 9.88, 7.69$ yielding $\epsilon_{13.2} (\theta) \leq 10^{-6}$, $\epsilon_{9.88} (\theta) \approx 5 \cdot 10^{-5}$, and $\epsilon_{7.69} (\theta) \approx  5 \cdot 10^{-4}$}}
 \label{fig:Poisson}
\end{figure}

\subsection[In practice]{The propagation condition in practice}\label{sec:approxPropCond}
\label{sec:propCondPraxis}

The propagation condition 
is based on the function~$\mathcal{Z}_{\lambda}$. This depends on the probability 
$\mathbb{P} \left( \overline{N}_i^{(k)} \mathcal{KL}(\tilde{\theta}_i^{(k)}(\lambda), \theta) > z \right)$ 
which cannot be calculated exactly. Therefore, in practice, we need an appropriate approximation. This can be achieved by the relative frequency of design points~$X_i \in \mathcal{X}$ with 
$\overline{N}_i^{(k)} \mathcal{KL}(\tilde{\theta}_i^{(k)}(\lambda), \theta) > z $ as we discuss in Definition~\ref{def:approx} and Lemma~\ref{lem:approx}. In order to avoid boundary effects, we restrict 
the approximation to the interior of the design space, that is to all points~$X_i \in \mathcal{X}$ where the final neighborhood~$U_i^{(k^*)} $ is not restricted by the boundaries of the considered compartment~$\lbrace X_i \rbrace_{i=1}^n$. This subset of~$\lbrace X_i \rbrace_{i=1}^n$ is denoted by~$\mathcal{X}^{0}$. Without loss of generality we assume that $\mathcal{X}^0 = \lbrace X_i \rbrace_{i=1}^{n_{0}}$ for some~$n_{0} < n$.

\begin{Def}[Approximation]\label{def:approx}
We consider the same setting as in Definition~\ref{def:propCond} and set
\[
	M_{\lambda}^{(k)}(z) := \lbrace X_i \in \mathcal{X}^{0} : \overline{N}_i^{(k)} \mathcal{KL}(\tilde{\theta}_i^{(k)}(\lambda), \theta) > z \rbrace.
\]
Then we define the following estimator
\begin{equation}\label{eq:estimator}
	\hat{p}_{\lambda}^{(k)}(z) := n_0^{-1} \sum_{i=1}^{n_{0}} \boldsymbol{1}_{M_{\lambda}^{(k)}(z)}(X_i)
\end{equation}
where~$\boldsymbol{1}$ denotes the indicator function with~$\boldsymbol{1}_M(x) =1$ if~$x \in M$ and~$\boldsymbol{1}_M(x) = 0$, else. 
\end{Def}

\begin{Lem}\label{prop:approx}\label{lem:approx}
We consider the same setting as in Definition~\ref{def:propCond} and suppose the conditions of Proposition~\ref{prop:propCond} to be satisfied. Then, it holds for each $j \in \lbrace 1,...,n_0 \rbrace $ that
\[
	\left| \mathbb{E} \left[ \hat{p}_{\lambda}^{(k)}(z) \right] - \mathbb{P} \left( \overline{N}_j^{(k)} \mathcal{KL} (\tilde{\theta}_{j}^{(k)}(\lambda), \theta ) > z \right) \right| 
	\leq \max \lbrace 2 e^{-z}, \epsilon \rbrace
\]
and
\begin{equation} \label{eq:propCond3}
	\mathrm{Var} \left[ \hat{p}_{\lambda}^{(k)}(z) \right] \leq \max \lbrace 2 e^{-z}, \epsilon \rbrace.
\end{equation}
\end{Lem}

\begin{proof}
It holds
\begin{eqnarray*}
	&&\left| \mathbb{E} \left[ \hat{p}^{(l)}(z) \right] - \mathbb{P} \left( \overline{N}_{j}^{(k)} \mathcal{KL} (\tilde{\theta}_{j}^{(k)}(\lambda), \theta ) > z \right) \right| \\
	&\leq& n_0^{-1} \sum_{i=1}^{n_0} \left| \mathbb{E} \left[ \boldsymbol{1}_{M_{(\lambda)}^{(k)}(z)}(X_i) \right] - \mathbb{P} \left( \overline{N}_{j}^{(k)} \mathcal{KL} (\tilde{\theta}_{j}^{(k)}(\lambda), \theta ) > z \right) \right| \\
	&\leq& \max_{i \in \lbrace 1,...,n_0 \rbrace} \left\{ \left| \mathbb{P} \left( \overline{N}_{i}^{(k)} \mathcal{KL} (\tilde{\theta}_{i}^{(k)}(\lambda), \theta ) > z \right) - \mathbb{P} \left( \overline{N}_{j}^{(k)} \mathcal{KL} (\tilde{\theta}_{j}^{(k)}(\lambda), \theta ) > z \right) \right| \right\}\\
	&\leq& \max_{i \in \lbrace 1,...,n_0 \rbrace} \mathbb{P} \left( \overline{N}_{i}^{(k)} \mathcal{KL} (\tilde{\theta}_{i}^{(k)}(\lambda), \theta ) > z \right)\\
	&\overset{\text{Prop.~\ref{prop:propCond}}}{\leq}& \max \lbrace 2 e^{-z}, \epsilon \rbrace.
\end{eqnarray*}
Furthermore, we get
\begin{eqnarray*}
	\mathrm{Var} \left[ \hat{p}_{\lambda}^{(k)}(z) \right] 
	&=& \left\| n_0^{-1} \sum_{i=1}^{n_{0}} \left(  \boldsymbol{1}_{M_{\lambda}^{(k)}(z)}(X_i) - \mathbb{E} \left[ \boldsymbol{1}_{M_{\lambda}^{(k)}(z)}(X_i) \right] \right) \right\|_{\mathbb{L}^2}^2 \\
	&\leq& \left( n_0^{-1} \sum_{i=1}^{n_{0}} \left\| \boldsymbol{1}_{M_{\lambda}^{(k)}(z)}(X_i) - \mathbb{E} \left[ \boldsymbol{1}_{M_{\lambda}^{(k)}(z)}(X_i) \right] \right\|_{\mathbb{L}^2} \right)^2 \\
	&\leq& \max_{i \in \lbrace 1,...,n_0 \rbrace} \mathrm{Var} \left[ \boldsymbol{1}_{M_{\lambda}^{(k)}(z)}(X_i) \right].
\end{eqnarray*}
Obviously, it holds for any random variable~$X$ with values in~$[0,1]$
that $\mathrm{Var}[X] \leq \mathbb{E}[X]$. 
By definition of~$M_{\lambda}^{(k)}(z)$ this yields
\begin{eqnarray*}
	\max_{i \in \lbrace 1,...,n_0 \rbrace} \mathbb{E} \left[ \boldsymbol{1}_{M_{\lambda}^{(k)}(z)}(X_i) \right]
	&=& \max_{i \in \lbrace 1,...,n_0 \rbrace} \mathbb{P} \left( \overline{N}_i^{(k)} \mathcal{KL}(\tilde{\theta}_i^{(k)}(\lambda), \theta) > z \right) \\
	&\overset{\text{Prop.~\ref{prop:propCond}}}{\leq}& \max \lbrace 2 e^{-z}, \epsilon \rbrace
\end{eqnarray*}
leading to Equation~(\ref{eq:propCond3}).
\end{proof}

\begin{Rem}\hspace{1 pt}
Theorem~\ref{thm:locPropNeuMS} provides a meaningful result only if $\epsilon := c_{\epsilon} n^{-q}$ with~$c_{\epsilon} > 0$ and~$q > 1$. We approximate the probability $\mathbb{P} \left( \overline{N}_i^{(k)} \mathcal{KL}(\tilde{\theta}_i^{(k)}(\lambda), \theta) > z \right)$ by the corresponding relative frequency~(\ref{eq:estimator}). This estimate can be calculated for $\epsilon \geq 1/n$ only. Additionally, it becomes instable if~$\epsilon$ is close to~$1/n$. 
In case of a regular design, the sample can be extended in a natural way allowing arbitrary sample sizes and as a consequence any~$\epsilon>0$. Otherwise, that is for random or irregular designs, we can achieve $\epsilon := c_{\epsilon} n^{-q}$ with~$c_{\epsilon} > 0$ and~$q > 1$ solely by application of the propagation condition on an artificial data set with~$m$ design points, where~$m \gg n$. In this case, one should evaluate carefully under which conditions the propagation condition generalizes from the artificial data set to the data set at hand.
\end{Rem}

\subsection[Generalization]{Generalization of the setting}\label{sec:genModel}

Assumption~(\ref{A1}) and hence the whole study were restricted to the case $\mathbb{E}_{\theta} \left[ T(Y) \right] = \theta$. Which modifications and additional assumptions are required in order to take the previous results over to the case where $ t(\theta) := \mathbb{E}_{\theta} \left[ T(Y) \right]$ is some invertible function? 


As mentioned in Remark~\ref{rem:A1}, $\mathbb{E}_{\theta} \left[ T(Y) \right] = \theta$ for all~$\theta \in \Theta$ can be achieved via reparametrization. Estimation of a parameter~$\vartheta$ with $t(\vartheta) := \mathbb{E}_{\vartheta} \left[ T(Y) \right] \neq \vartheta$ can still be done for invertible functions~$t(.)$, setting $\tilde{\vartheta}_i^{(k)} := t^{-1} ( \tilde{\theta}_i^{(k)} ) $ for all $i \in \lbrace 1,...,n \rbrace$ and $k \in \lbrace 0,...,k^* \rbrace$, where~$\tilde{\theta}_i^{(k)}$ denotes the adaptive estimator resulting from Algorithm~\ref{algorithm}. Hence the algorithm remains unmodified! We will see that all results in Sections~\ref{sec:theory}, \ref{sec:Lambda} and~\ref{sec:propCondPraxis} remain valid if~$t(\vartheta)$ is linear in~$\vartheta$. This generalizes our previous results to the Gamma, Erlang, Rayleigh, Binomial, and negative Binomial distributions, see Appendix~\ref{app:A1}.

\setcounter{Ass}{0}
\renewcommand{\theAss}{A\arabic{Ass}g}
\begin{Ass}[Parametrized exponential family model]\label{A1g}
$\mathcal{P}^{(t)} =(\mathbb{P}^{(t)}_{\vartheta}, \vartheta \in \Theta)$ is an exponential family 
with a compact and convex parameter set~$\Theta$ and strictly monotone functions $C_t, B_t \in C^2\left( \Theta, \mathbb{R} \right)$ such that
\[
	p_t(y, \vartheta) := d \mathbb{P}^{(t)}_{\vartheta} / d \mathbb{P} (y) = p(y) \exp \left[ T(y) C_t (\vartheta) - B_t (\vartheta) \right], \qquad \vartheta \in \Theta,
\]
where $T: \mathcal{Y} \to \mathbb{R} $ and~$p(y)$ is some non-negative function on~$\mathcal{Y}$. For the parameter~$\vartheta$ it holds
\[
	\int p_t(y, \vartheta) \mathbb{P}(dy) = 1 \quad \text{ and } \quad 
	\mathbb{E}^{(t)}_{\vartheta} \left[ T(Y) \right] = \frac{B_t' (\vartheta)}{C_t' (\vartheta)} =: t(\vartheta),
\]
where $t: \Theta \to \Theta$ denotes an invertible and continuously differentiable function.
\end{Ass}

\begin{Cor}\label{prop:A1g}
Let Assumption~(\ref{A1g}) be satisfied.
Reparametrization with $\theta := t(\vartheta)$ yields 
\begin{equation}\label{eq:KLt-g1}
	\mathcal{KL} \left( \vartheta_1, \vartheta_2 \right) 
	= \mathcal{KL} \left( \theta_1, \theta_2 \right) \qquad \text{  for all } \vartheta_1, \vartheta_2 \in \Theta.
\end{equation}
If~$t(\vartheta)$ is linear in~$\vartheta$, then it follows for the adaptive estimator $\tilde{\vartheta} := t^{-1} \left( \tilde{\theta} \right)$ that
\begin{equation}\label{eq:KLt-g}
	\mathcal{KL}(\tilde{\vartheta}, \mathbb{E} \tilde{\vartheta}) = \mathcal{KL}(\tilde{\theta}, \mathbb{E} \tilde{\theta}).
\end{equation}
\end{Cor}

If~$t(\vartheta)$ is linear in~$\vartheta$ and if the adaptive estimator of $\vartheta$ is defined by $\tilde{\vartheta}_i^{(k)} := t^{-1} ( \tilde{\theta}_i^{(k)} ) $ for all $i \in \lbrace 1,...,n \rbrace$ and $k \in \lbrace 0,...,k^* \rbrace$, then it follows from Corollary~\ref{prop:A1g} that all previous results remain valid under Assumption~(\ref{A1g}), where the formulations of the propagation condition and Assumptions~(\ref{AEst}) and~(\ref{AS}) can be adapted to the generalized setting via~$\vartheta = t^{-1}(\theta)$. 

The exponential bound~(PS~\ref{PS 2.1}) 
is the only result, where we really need that $ \mathbb{E}_{\theta} \left[ T(Y) \right] = \theta$. All other proofs could be shown directly, i.e. without reparametrization by~$\theta = t(\vartheta)$. 
Here, the convexity of the Kullback-Leibler divergence w.r.t. the first argument holds if
\[
	\tfrac{\partial^2}{\partial \theta^2} \, \mathcal{KL} \left( \theta, \theta' \right)
	= t'' (\theta) \left[ C(\theta) - C(\theta') \right] + t'(\theta) C'(\theta) > 0.
\]
Then, the proof of~(PS~\ref{PS 2.1}) can be generalized supposing Assumption~(\ref{A1g}) and 
\[
	D_t'(\hat{\nu}) \geq \sum_{j=1}^n w_j D_t'(\nu_j) \quad \text{ with }
	\nu := C_t (\vartheta), D_t(\nu) := B_t(\theta), \text{ and } \hat{\vartheta} := \sum_{j=1}^n w_j \mathbb{E} \left[ T(Y_j) \right].
\]
However, for many parametric families this inequality is violated. That is why we prefer to apply~(PS~\ref{PS 2.1}) in its original form, where $\mathbb{E}_{\theta} \left[ T(Y) \right] = \theta $, and generalize the exponential bound afterwards via Equation~(\ref{eq:KLt-g}).

\section{Conclusion}

This study provides theoretical properties for a simplified version of the Pro\-pa\-ga\-tion-Sepa\-ra\-tion approach, where the memory step is removed from the algorithm. In particular, we have verified the following results, which may help for a better understanding of the procedure.
\begin{itemize}
\item In Section~\ref{sec:propCond}, we introduced an advanced parameter choice strategy for the adaptation bandwidth~$\lambda$. Its dependence on the unknown parameter function is analyzed in Section~\ref{sec:Lambda} showing for the first time theoretical and numerical results that justify the propagation condition. 
\item This parameter choice yields strong results on propagation and stability of estimates for piecewise constant functions with sharp discontinuities, see Section~\ref{sec:theory}. 
\item Finally, we gave some more details concerning the application of the propagation condition in practice, see Section~\ref{sec:propCondPraxis}, and a generalization of the assumed setting, Section~\ref{sec:genModel}.
\item In Remark~\ref{rem:AEst}, we proposed a slight modification of the algorithm providing Assumption~(\ref{AEst}) on which the results in Section~\ref{sec:theory} were partially based.
\end{itemize}
The behavior of the algorithm and hence the achievable quality of estimation depend mainly on the extension of the homogeneous compartments, on the smoothness of the parameter function~$\theta(.)$,
and via the adaptation bandwidth~$\lambda$ on 
the parametric family $\mathcal{P} = \lbrace \mathbb{P}_{\theta} \rbrace_{\theta \in \Theta}$ of probability distributions. Our theoretical results give an intuition of the interplay of propagation and separation during iteration. 
Future research may concentrate on the case of model misspecification in order to justify the heuristic observations in Section~\ref{sec:Heurstics}, mathematically.

\appendix

\section[Reminder]{Exponential bound and technical lemma}\label{app:reminder}

We remind of two results which have been proven in \citep[Lemma 5.2, Theorem 2.1]{PoSp05}. 

\begin{PS}[Technical Lemma]\label{PS 5.2}
Under Assumption~(\ref{A1}) it holds
\[
	\mathcal{KL}^{1/2} \left( \theta_0, \theta_m \right)
	\leq \varkappa \, \sum_{l=1}^m \mathcal{KL}^{1/2} \left( \theta_{l-1}, \theta_l \right)
\]
for any sequence~$\theta_0, \theta_1,..., \theta_m \in \Theta_{\varkappa}$, where~$\varkappa > 0$ is as in Lemma~\ref{lem:A1}.
\end{PS}

\begin{PS}[Exponential bound]\label{PS 2.1}
If~$\theta(.) \equiv \theta$ and Assumption~(\ref{A1}) is satisfied then it holds
\[
	\mathbb{P} \left( N \, \mathcal{KL}(\overline{\theta}, \theta) > z \right) \leq 2 e^{-z}, \qquad \forall \, z > 0,
\]
where $N := \sum_{j=1}^n w_{j}$ and $\overline{\theta} := \sum_{j=1}^n w_{j} T(Y_j) / N$ with given weights $w_j \in [0,1]$.
\end{PS}

\newpage

\section[Examples]{Examples for parametric families}\label{app:A1}

\begin{table}[hp]
\begin{scriptsize}
\begin{center}
\begin{tabular}{l|cccccc}
$\mathcal{P}$, $\mathrm{support}(f_{\vartheta})$ & $\Theta$ & $p(y)$ & $T(y)$ & $C_t(\vartheta)$ & $B_t(\vartheta)$ & $\mathbb{E}_{\vartheta} \left[ T(Y) \right]$ 
\\ \hline
&&&&&& \\
$\mathcal{N} (\vartheta, \sigma^2)$ & $\mathbb{R}$ & $\dfrac{e^{-y^2/(2 \sigma^2)}}{\sqrt{2 \pi \sigma^2}}$ & $y$ & $\dfrac{\vartheta}{\sigma^2}$ & $\dfrac{\vartheta^2}{2 \sigma^2}$ & $\vartheta$ 
\\
$\qquad y \in \mathbb{R}$ &&&&&& \\
$\mathcal{N} (0, \vartheta)$ & $(0, \infty)$ & $\dfrac{1}{\sqrt{2 \pi}}$ & $y^2$ & $- \dfrac{1}{2 \vartheta}$ & $ \dfrac{\ln \vartheta}{2}$ & $\vartheta$ 
\\
$\qquad y \in \mathbb{R}$ &&&&&&\\
$\log \mathcal{N} (\vartheta, \sigma^2)$ & $(0, \infty)$ & $\dfrac{e^{-(\ln y)^2/(2 \sigma^2)}}{y \sqrt{2 \pi \sigma^2}}$ & $ \ln y $ & $\dfrac{\vartheta}{\sigma^2}$ & $\dfrac{\vartheta^2}{2 \sigma^2}$ & $ \vartheta$ 
\\ 
$\qquad y \in (0, \infty)$ &&&&&& \\
$\Gamma(p, \vartheta) $ & $(0, \infty)$ & $\dfrac{y^{p-1}}{\Gamma(p)}$ & $ y $ & $ - \dfrac{1}{\vartheta} $ & $ p \ln \vartheta $ & $ p \, \vartheta $ 
\\
$\qquad y \in (0, \infty)$ &&&&&& \\
$\mathrm{Exp}\left( \dfrac{1}{\vartheta} \right)$ &  $(0, \infty)$ & $ 1 $ & $ y $ & $ - \dfrac{1}{\vartheta} $ & $ \ln \vartheta $ & $ \vartheta $ 
\\
$\qquad y \in [0, \infty)$ &&&&&& \\
$\mathrm{Erlang} \left( n, \dfrac{1}{\vartheta} \right)$ &  $(0, \infty)$ & $\dfrac{y^{n-1}}{(n-1)!}$ & $ y $ & $ - \dfrac{1}{\vartheta} $ & $ n \ln \vartheta $ & $n \, \vartheta $ 
\\
$\qquad y \in [0, \infty)$ &&&&&& \\
$\mathrm{Rayleigh} (\vartheta) $ & $(0, \infty)$ & $ y $ & $ y^2 $ & $ - \dfrac{1}{2 \vartheta^2} $ & $ 2 \ln \vartheta $ & $ 2 \vartheta^2 $ 
\\
$\qquad y \in [0, \infty)$ &&&&&& \\
$\mathrm{Weibull} (\vartheta, k) $ & $(0, \infty)$ & $ k y^{k-1} $ & $ y^k $ & $ - \dfrac{1}{\vartheta^k} $ & $ k \ln \vartheta $ & $ \vartheta^k $ 
\\
$\qquad y \in [0, \infty)$ &&&&&& \\
$k Y / \vartheta \sim \chi^2 (k)$ & $(0, \infty)$ & $\dfrac{k^{k/2}	 y^{k/2-1}}{2^{k/2} \Gamma \left( k / 2 \right)}$ & $y$ & $- \dfrac{k}{2 \vartheta}$ & $ \dfrac{k \ln \vartheta}{2} $ & $ \vartheta $ 
\\
$\qquad y \in [0, \infty)$ &&&&&&  \\
$\mathrm{Pareto}(x_m, \vartheta)$ & $(1, \infty)$ & $ \dfrac{1}{y} $ & $ \ln \left( \dfrac{y}{x_m} \right) $ & $ - \vartheta $ & $ - \ln \left( \vartheta \right) $ & $ \dfrac{1}{\vartheta} $ \\
$\qquad y \in [x_m, \infty) $ &&&&&
\end{tabular}
\end{center}
\end{scriptsize}
\vspace{12 pt}
\caption{One-parametric exponential families which satisfy Assumption~(\ref{A1g}): 
Continuous distributions}
\label{tab:contDistr}
\end{table}

\begin{table}[hp]
\begin{scriptsize}
\begin{center}
\begin{tabular}{l|cccccc}
$\mathcal{P}$, $\mathrm{support}(f_{\vartheta})$ & $\Theta$ & $p(y)$ & $T(y)$ & $C_t(\vartheta)$ & $B_t(\vartheta)$ & $\mathbb{E}_{\vartheta} \left[ T(Y) \right]$ \\ \hline
&&&&&& \\
$\mathrm{Poiss} (\vartheta) $ & $(0, \infty)$ & $1/k!$ & $ k $ & $ \ln \vartheta $ & $ \vartheta $ & $ \vartheta $  \\
$ \qquad y := k \in \mathbb{N}$ &&&&& \\
$\mathrm{Bin}(n, \vartheta)$ & $(0,1]$ & $\left(\begin{array}{c} n \\ k \end{array}\right)$ & $k$ & $\ln \left( \dfrac{\vartheta}{1 - \vartheta} \right)$ & $- n \ln (1 - \vartheta) $ & $ n \, \vartheta $ \\
$\qquad y:= k \in \lbrace 0,1,...,n \rbrace$ &&&&& \\
$\mathrm{NegativeBin}(r, \vartheta)$ & $(0,1]$ & $\left(\begin{array}{c} k+r-1 \\ k \end{array}\right)$ & $k$ & $\ln \vartheta$ & $- r \ln (1 - \vartheta) $ & $ \dfrac{r \vartheta}{1 - \vartheta} $ \\
$\qquad y:= k \in \mathbb{N}$ &&&&& \\
$\mathrm{Bernoulli}(\vartheta)$ & $(0,1]$ & $1$ & $k$  & $\ln \left( \dfrac{\vartheta}{1 - \vartheta} \right)$ & $- \ln (1 - \vartheta) $ & $ \vartheta $ \\
$\qquad y:= k \in \lbrace 0,1 \rbrace$ &&&&&&
\end{tabular}
\end{center}
\end{scriptsize}
\vspace{12 pt}
\caption{One-parametric exponential families which satisfy Assumption~(\ref{A1g}): 
Discrete distributions}
\label{tab:discrDistr}
\end{table}

\section*{Acknowledgements}

This work was partially supported by the 
Stiftung der Deutschen Wirtschaft (SDW). The authors would like to thank J\"{o}rg Polzehl, Vladimir Spokoiny and Karsten Tabelow (WIAS Berlin) for helpful discussions.

\bibliographystyle{plainnat}
\bibliography{PS06-revision}

\end{document}